\documentclass[11pt]{article}
\usepackage[margin=3cm]{geometry}
\usepackage{graphicx}
\usepackage{amsfonts}
\usepackage{amsmath,amsthm}
\usepackage{amssymb}
\usepackage{enumitem}
\usepackage[]{todonotes} 
\usepackage{comment}
\usepackage{hyperref}
\usepackage{algorithm}
\usepackage{algorithmic}
\usepackage{authblk}
\usepackage{array}

\def\NP{\textsf{NP}}
\def\PP{PP}
\def\sNPc{sNPc}

\newcommand{\tin}{\mathsf{tree} \textnormal{-} \alpha}
\newcolumntype{?}[1]{!{\vrule width #1}}

\newtheorem{definition}{Definition}
\newtheorem{theorem}{Theorem}

\newtheorem{lemm}[theorem]{Lemma}
\newtheorem{coro}[theorem]{Corollary}

\usepackage{tabularx, environ}

\makeatletter

\newcolumntype{\expand}{}
\long\@namedef{NC@rewrite@\string\expand}{\expandafter\NC@find}

\NewEnviron{fancyproblem}[2][]{%
	\def\problem@arg{#1}%
	\def\problem@framed{framed}%
	\def\problem@lined{lined}%
	\def\problem@doublelined{doublelined}%
	\ifx\problem@arg\@empty%
	\def\problem@hline{}%
	\else%
	\ifx\problem@arg\problem@doublelined%
	\def\problem@hline{\hline\hline}%
	\else%
	\def\problem@hline{\hline}%
	\fi%
	\fi%
	\ifx\problem@arg\problem@framed%
	\def\problem@tablelayout{|>{\bfseries}lX|c}%
	\def\problem@title{\multicolumn{2}{|l|}{%
			\raisebox{-\fboxsep}{\textsc{\large #2}}%
	}}%
	\else
	\def\problem@tablelayout{>{\bfseries}lXc}%
	\def\problem@title{\multicolumn{2}{l}{%
			\raisebox{-\fboxsep}{\textsc{\large #2}}%
	}}%
	\fi%
	\bigskip\par\noindent%
	\begin{tabularx}{\textwidth}{\expand\problem@tablelayout}%
		\problem@hline%
		\problem@title\\[2\fboxsep]%
		\BODY\\\problem@hline%
	\end{tabularx}%
	\medskip\par%
}
\makeatother

\title{Fair Allocation Algorithms for Indivisible Items\\ under Structured Conflict Constraints}

\author[1,2]{Nina Chiarelli}
\author[1,2]{Matja\v z Krnc}
\author[1,2]{Martin Milani\v c}
\author[3]{Ulrich Pferschy}
\author[4]{Joachim Schauer}

\affil[1]{FAMNIT, University of Primorska, Glagolja\v ska 8, 6000 Koper, Slovenia}
\affil[2]{IAM, University of Primorska, Muzejski trg 2, 6000 Koper, Slovenia}
\affil[3]{Department of Operations and Information Systems, University of Graz,\\ Universitaetsstrasse 15/E3, 8010 Graz, Austria}
\affil[4]{FH JOANNEUM, Werk-VI-Straße 46, 8605 Kapfenberg, Austria}
\begin{document}

\maketitle

	\begin{abstract} 
	We consider the fair allocation of indivisible items to several agents with additional conflict constraints.
	These are represented by a conflict graph where each item corresponds to a vertex of the graph and edges in the graph represent incompatible pairs of items which should not be allocated to the same agent.
	This setting combines the issues of \textsc{Partition} and \textsc{Independent Set} and can be seen as a partial coloring of the conflict graph.
In the resulting optimization problem each agent has its own valuation function for the profits of the items. 
	We aim at maximizing the lowest total profit obtained by any of the agents.
	In a previous paper this problem was shown to be strongly \NP-hard for several well-known graph classes, e.g., bipartite graphs and their line graphs.
    On the other hand, it was shown that pseudo-polynomial time algorithms exist for the classes of chordal graphs, cocomparability graphs, biconvex bipartite graphs, and graphs of bounded treewidth.
	In this contribution we extend this line of research by developing pseudo-polynomial time algorithms that solve the problem for the class of convex bipartite conflict graphs, graphs of bounded clique-width, and graphs of bounded tree-independence number.
	The algorithms are based on dynamic programming and also permit fully polynomial-time approximation schemes (FPTAS).
	\end{abstract}

\bigskip
\noindent
\textbf{Keywords:}
    fair division;
    conflict graph; 
    partial coloring; 
    convex bipartite graphs; 
    bounded clique-width;
    bounded tree-independence number.

\bigskip
\noindent
\textbf{Mathematics Subject Classification (2020):} 
90C27,
05C85,
90C47,
91B32,
90C39,
68Q25,
68W25.

\newpage

\section{Introduction}
\label{sec:intro}

Optimization problems arising from allocation of resources or apportionment of goods often contain conditions on the feasibility of a solution which are best represented by a graph.
In particular, incompatibilities between pairs of objects can be modeled in a natural way by means of a {\em conflict graph}. 
There is a wide range of literature where conflict graphs are added to combinatorial optimization problems leading to new problems with feasible solutions consisting of objects whose graph representations are independent sets in the conflict graph, e.g., knapsack problems (\cite{pfsch09,pfsch17}), bin packing (\cite{mimt10}), scheduling (e.g.,~\cite{boja95,raey09}), transportation (\cite{Santos19}) and problems on graphs (e.g.,~\cite{dpsw11}).

In general, optimizing over independent sets will make most optimization problems computationally hard. 
Therefore, a widely pursued research direction consists in the identification of special graph classes where the implied optimization problem exhibits a more benevolent behavior.

We consider a very natural allocation problem where a set of $n$ indivisible goods or items has to be distributed among $k$ agents.
Each agent has his or her own profit function over the set of items.
The profits of items are additive.
To obtain a {\em fair allocation} of items the minimal total profit obtained by any of the agents should be maximized.
Such a maxi-min criterion is probably the most natural fairness measure and can be found frequently in Computational Social Choice.
This problem is well-known as the \emph{Santa Claus} problem (see~\cite{MR2277128})
where presents \hbox{(items)} should be distributed to kids (agents) so that even the least happy kid is as happy as possible.
It is easy to see by a reduction from \textsc{Partition} that this problem is weakly \NP-hard already for $k=2$, even if all $k$ profit functions are identical.
Thus, the best we can hope for, from a computational complexity point of view, are pseudo-polynomial algorithms for restricted classes of input instances.

We consider restrictions on input instances giving rise to special classes expressed in terms of properties of the corresponding conflict graph permitting an incompatibility relation between pairs of items.
Two vertices are joined by an edge if the two corresponding items do not fit together well, e.g., {being substitutes for each other.}
In the resulting optimization problem we are looking for the most fair solution where the allocation to every agent constitutes an independent set in the conflict graph.

\medskip
In the present paper we continue the line of research started in our previous work~\cite{ourarx} where the fair allocation problem with conflicts was first introduced and motivated, and the relation to existing literature was laid out in detail.
In that paper several graph classes were classified as belonging to one of the two sides of the complexity divide, i.e.,\ the problem, when restricted to graphs from the class, is either strongly \NP-hard or permits a pseudo-polynomial algorithm.
In the current paper we add three highly relevant types of graphs to the positive side and thus proceed to narrow the gap between hard and tractable cases.
As in the previous paper~\cite{ourarx}, we consider the number of agents $k$ as a constant.

\begin{sloppypar}
Obviously, the fair division problem with conflicts coincides with the strongly \NP-hard \textsc{Weighted Independent Set} problem (WIS) for the case of a single agent ($k=1$). 
Thus, it only makes sense to consider graph classes where WIS can be solved in polynomial time.
This immediately points to perfect graphs~\cite{MR936633}. 
However, our allocation problem was shown to be strongly \hbox{\NP-hard} for several subclasses of perfect graphs as conflict graphs~\cite{ourarx}, among them the prominent class of bipartite graphs.
On the positive side, we could show that {\em biconvex bipartite graphs} still permit pseudo-polynomial solution algorithms.
In this paper, we manage to push the bar further up by presenting a fairly involved pseudo-polynomial time algorithm for {\em convex bipartite graphs} (see Section~\ref{sec:def} for definitions).
\end{sloppypar}

Outside the area of perfect graphs, we construct  pseudo-polynomial algorithms for conflict graphs of {\em bounded clique-width} or {\em bounded tree-independence number.}
Both of these results extend the only previously known cases of tractable non-perfect graphs, namely graphs of bounded treewidth.

\subsection{Problem statement and results}
\label{sec:problem-statement-and-results}

For a formal definition of our problem let $V$ be a set of items with $|V|=n$, let $k$ be a positive integer, and let $p_1,\ldots,p_k: V \to \mathbb{Z}_+$ be $k$ profit functions.
{An \emph{ordered \hbox{$k$-partition}} of $V$ is a sequence $(X_1,\ldots,X_k)$ of $k$ pairwise disjoint subsets of $V$ such that $\bigcup_{i = 1}^k X_i = V$.}
Given profit functions $p_1,\dots, p_k$, the \emph{satisfaction level} of an ordered $k$-partition $(X_1,\ldots,X_k)$ of $V$ is defined as the minimum over the $k$ individual profits $p_j(X_j) := \sum_{v\in X_j}p_j(v)$, $j\in \{1,\ldots, k\}$.
The standard fair division problem is defined as follows.

\medskip
\noindent\parbox{0.82\linewidth}{\noindent
	{\sc Fair $k$-Division of Indivisible Items}\\[1.2ex]
	\begin{tabular*}{.8\textwidth}{ll}
		\noindent{\bf Input:} & A set $V$ of $n$ items, $k$ profit functions $p_1,\ldots,p_k: V \to \mathbb{Z}_+$.\\[0.5ex]
		\noindent{\bf Task:} & Compute an ordered $k$-partition of $V$ with maximum satisfaction  level.
	\end{tabular*}
}
    
\smallskip
A \emph{conflict graph} $G = (V,E)$ on the set $V$ of items
represents incompatibilities between pairs of items.
If two items $i$ and $j$ are joined by an edge $\{i,j\} \in E$,
then $i$ and $j$ should not be included in the same subset of the partition.
Conflict-free allocation of items immediately leads to (partial) colorings of the conflict graph (cf.~Berge~\cite{MR989117} and de Werra~\cite{MR1097650}).

\begin{definition}
	A \emph{partial $k$-coloring} of a graph $G$ is a sequence $(X_1,\ldots,X_k)$ of $k$ pairwise disjoint independent sets in $G$.
\end{definition}

Obviously, any partial $k$-coloring of a graph $G$ corresponds to a $k$-division of the items respecting the incompatibilities, since every set in the partition is an independent set in $G$.
Note that an optimal partial $k$-coloring $(X_1,\ldots,X_k)$ does not necessarily select all vertices from $V$, i.e., some items may remain unassigned (which can be ruled out for the problem without conflict relations).
Taking the $k$ profit functions into account we can define for each partial $k$-coloring $c = (X_1,\ldots, X_k)$ of a graph $G$ 
a $k$-tuple $(p_1(X_1),\ldots, p_k(X_k))$ as the \emph{profit profile} of $c$.
Then the \emph{satisfaction level} of $c$ is given by $\min_{j=1}^k \{p_j(X_j)\}$, i.e., the minimum profit of a profile.
In this paper the following resulting optimization problem is considered:

\medskip
\noindent\parbox{0.82\linewidth}{\noindent
	{\sc Fair $k$-Division Under Conflicts}\\[1.2ex]
	\begin{tabular*}{.9\textwidth}{ll}
		\noindent{\bf Input:} & A graph $G = (V,E)$, $k$ profit functions $p_1,\ldots,p_k: V \to \mathbb{Z}_+$.\\[0.5ex]
		\noindent{\bf Task:} & Compute a partial $k$-coloring of $G$ with maximum satisfaction level.
	\end{tabular*}
}

\medskip\noindent
We will also refer to the decision version of the problem: 
given a target value $q \in \mathbb{Z}_+$, does there exist a partial $k$-coloring of $G$ with satisfaction level at least $q$?

\medskip
Even without conflicts the plain problem
{\sc Fair $k$-Division of Indivisible Items}
is weakly \NP-hard for any constant $k\geq 2$ and strongly \NP-hard for $k$ being part of the input.
This holds even for $k$ identical profit functions
(see the discussion in~\cite{ourarx}).
Thus, pseudo-polynomial algorithms for 
{\sc Fair $k$-Division Under Conflicts}
can only be developed for constant $k$.

\medskip
In our preceding paper~\cite{ourarx} we derived the following results.\footnote{We will only state the definitions of those graph classes that play a role in this paper.} The decision version of {\sc Fair $k$-Division Under Conflicts} is strongly \NP-complete for any fixed $k\geq 2$ if the conflict graph is a {\em bipartite graph} or a {\em line graph of a bipartite graph}.
Both results imply that the problem is hard on {\em perfect graphs}, where {\sc Weighted Independent Set} is still polynomial.

On the other hand we established pseudo-polynomial solution algorithms for the following graph classes: biconvex bipartite graphs, cocomparability graphs, chordal graphs, and graphs of bounded treewidth.

It was also shown that all these algorithms, which rely on dynamic programs of pseudo-polynomial size, permit the construction of fully polynomial time approximation schemes (FPTAS).

\medskip
In this paper we give, for any fixed positive integer $k$, a pseudo-polynomial algorithm for {\sc Fair $k$-Division Under Conflicts} in the class of {\em convex bipartite graphs} in Section~\ref{sec:convex}, thereby answering a question from~\cite{ourarx}.
Note that this approach is completely different from the previous algorithm for the {\em biconvex} bipartite case.
Secondly, we construct a pseudo-polynomial algorithm for {\sc Fair $k$-Division Under Conflicts} on conflict graphs of bounded clique-width in Section~\ref{sec:bcw}.
Since graph classes of bounded treewidth have bounded clique-width but not vice versa, this is a strict generalization of the analogous result on bounded treewidth from~\cite{ourarx}.
Thirdly, we provide another generalization of the result for bounded treewidth, which at the same time also generalizes the solution for chordal graphs, by adapting the bounded treewidth algorithm to graphs with bounded tree-independence number in Section~\ref{sec:tin}.

The presented algorithms explore the structural properties of the respective graph classes to generate suitable sets of profit profiles as states in a dynamic programming approach.
By the same reasoning as in Section~4 of \cite{ourarx} we can also conclude that each of the dynamic programming algorithms laid out for convex bipartite graphs, graphs of bounded clique-width, and graphs of bounded tree-independence number also leads to a fully polynomial time approximation scheme (FPTAS).

\smallskip
Pointing out further open questions, we would like to mention 
{\em $\chi$-bounded graph classes} (see, e.g.,~\cite{MR4174126}), which are graph classes $\mathcal{G}$ closed under induced subgraphs for which there exists a function $f$ bounding from above the chromatic number of each graph $G\in \mathcal{G}$ in terms of its clique number.
These can be seen as a simultaneous extension of the class of perfect graphs, as well as of graph classes of bounded clique-width (see~\cite{DVORAK2012}) and graphs of bounded tree-independence number (see~\cite{dallard2022firstpaper}).
Narrowing the gap of computational complexity between $\chi$-bounded graph classes and graphs of bounded clique-width, resp.~bounded tree-independence number, would be an interesting challenge.
An overview of the state of knowledge for various graph classes is given in Figure~\ref{fig:Hasse-new}. 

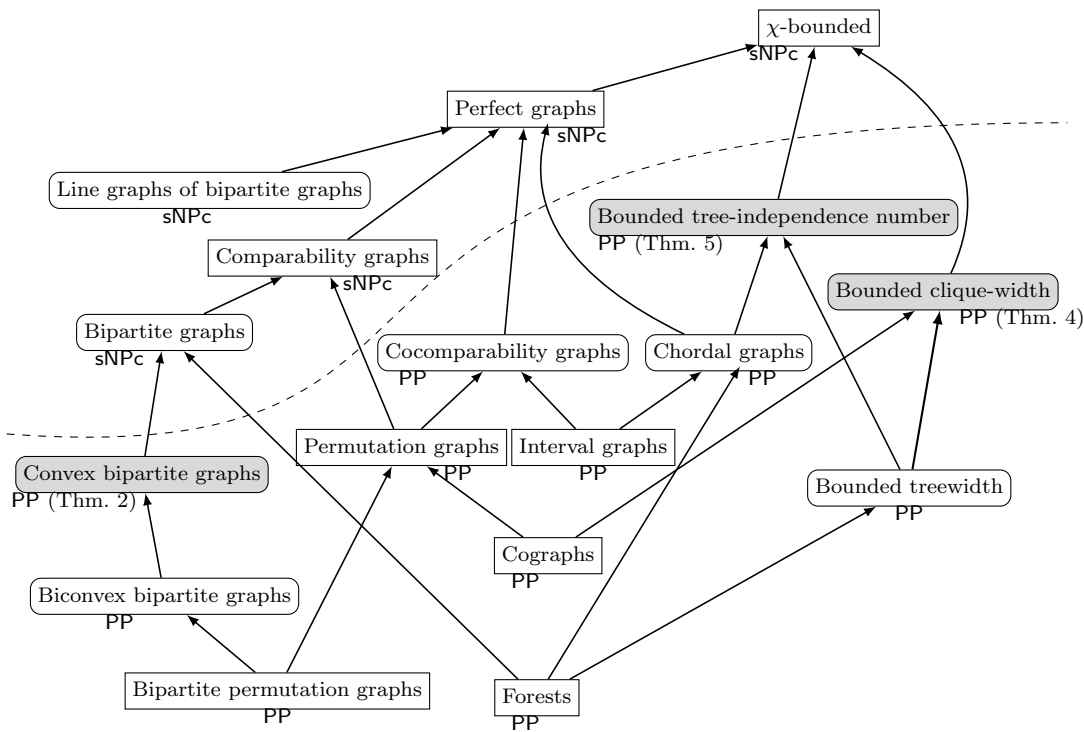
\begin{figure}[h]
	\centering
	\begin{tikzpicture}[xscale=3,yscale=1.8]
	\tikzset{roundedbox/.style={draw,rectangle,rounded corners}}
	\tikzset{cornerbox/.style={draw,rectangle}}
	\scriptsize
	\node[cornerbox] (bipartitePermutation) at (-0.7,-0.4) {\mathstrut Bipartite permutation graphs};
	\node at (-0.7,-0.6) {\mathstrut \textsf{\PP}};
	\node[roundedbox] (biconvexBipartite) at (-1.2,0.3) 
	{\mathstrut Biconvex bipartite graphs};
	\node at (-1.4,0.1) {\mathstrut \textsf{\PP} 
	};
	\node[roundedbox, fill=gray!30] (convexBipartite) at (-1.3,1.2) 
	{\mathstrut Convex bipartite graphs};
	\node at (-1.6,1.0) {\mathstrut \textsf{\PP} (Thm.~\ref{thm:convex})
	};
	\node[roundedbox] (bipartite) at (-1.2,2.25) {\mathstrut Bipartite graphs};
	\node at (-1.4,2.05) {\mathstrut \textsf{\sNPc} 
	};
	\node[cornerbox] (permutation) at (-0.15,1.4) {\mathstrut Permutation graphs};
	\node at (0.1,1.2) {\mathstrut \textsf{\PP}};
	\node[cornerbox] (interval) at (0.7,1.4) {\mathstrut Interval graphs};
	\node at (0.7,1.2) {\mathstrut \textsf{\PP}};
	\node[roundedbox] (cocomparability) at (0.3,2.1) {\mathstrut Cocomparability graphs};
    \node at (-0.1,1.9) {\mathstrut \textsf{\PP} 
    };
    \node[roundedbox] (chordal) at (1.3,2.1) {\mathstrut Chordal graphs};
    \node at (1.45,1.9) {\mathstrut \textsf{\PP} 
    };
    \node[cornerbox] (cographs) at (0.5,0.6) {\mathstrut Cographs};
	\node at (0.4,0.4) {\mathstrut \textsf{\PP}};
	
	\node[cornerbox] (comparability) at (-0.5,2.8) {\mathstrut Comparability graphs};
	\node at (-0.3,2.6) {\mathstrut \textsf{\sNPc}};
	\node[cornerbox] (perfect) at (0.4,3.9) {\mathstrut Perfect graphs};
	\node at (0.65,3.7) {\mathstrut \textsf{\sNPc}};
	\node[roundedbox] (lineBip) at (-1.0,3.3) {\mathstrut Line graphs of bipartite graphs};
	\node at (-1.1,3.1) {\mathstrut \textsf{\sNPc} 
	};
	\node[roundedbox] (bddtw) at (2.1,1.1) {\mathstrut Bounded treewidth};
	\node at (2.1,0.9) {\mathstrut \textsf{\PP} 
	};
        \node[roundedbox, fill=gray!30] (bdtin) at (1.5,3.1) {\mathstrut Bounded tree-independence number};
 	\node at (1,2.9){\mathstrut \textsf{\PP} (Thm.~\ref{thm:boundedtree-alpha})};
	\node[roundedbox, fill=gray!30] (bddcw) at (2.25,2.55) {\mathstrut Bounded clique-width};
 	\node at (2.6,2.35){\mathstrut \textsf{\PP} (Thm.~\ref{thm:boundedcliquewidth})};
	\node[cornerbox] (chibnd) at (1.7,4.5) {\mathstrut $\chi$-bounded};
  	\node at (1.5,4.3){\mathstrut\textsf{\sNPc}};
 	
	\node[cornerbox] (forests) at (0.45,-0.45) {\mathstrut Forests};
    \node at (0.4,-0.65) {\mathstrut \textsf{\PP}};
	\draw[semithick,-latex,font=\sffamily] (perfect) to (chibnd);
	\draw[semithick,-latex,bend right,font=\sffamily] (bddcw) to (chibnd);	
	\draw[semithick,-latex,font=\sffamily] (bipartitePermutation) to (biconvexBipartite);
	\draw[semithick,-latex,font=\sffamily] (biconvexBipartite) to (convexBipartite);
	\draw[semithick,-latex,font=\sffamily] (convexBipartite) to (bipartite);
	\draw[semithick,-latex,font=\sffamily] (bipartite) to (comparability);
	\draw[semithick,-latex,font=\sffamily] (interval) to (cocomparability);
	\draw[semithick,-latex,font=\sffamily] (interval) to (chordal);
	\draw[semithick,-latex,bend left,font=\sffamily] (chordal) to (0.5,3.8);
	\draw[semithick,-latex,font=\sffamily] (permutation) to (comparability);
	\draw[semithick,-latex,font=\sffamily] (permutation) to (cocomparability);
	\draw[semithick,-latex,font=\sffamily] (comparability) to (perfect);
	\draw[semithick,-latex,font=\sffamily] (cocomparability) to (perfect);
	\draw[semithick,-latex,font=\sffamily] (bipartitePermutation) to (permutation);
	\draw[semithick,-latex,font=\sffamily] (lineBip) to (perfect);
	\draw[semithick,-latex,font=\sffamily] (cographs) to (permutation);
	\draw[semithick,-latex,font=\sffamily] (cographs) to (bddcw);
	\draw[semithick,-latex,font=\sffamily] (forests) to (1.35,2);
	\draw[semithick,-latex,font=\sffamily] (forests) to (bipartite);
	\draw[thick,-latex,font=\sffamily] (bddtw) to (bddcw);l
	\draw[semithick,-latex,font=\sffamily] (forests) to (bddtw);
        \draw[semithick,-latex,font=\sffamily] (chordal) to (bdtin);
        \draw[semithick,-latex,font=\sffamily] (bddtw) to (bdtin);
        \draw[semithick,-latex,font=\sffamily] (bdtin) to (chibnd);
	\draw[dashed] (-1.9,1.5) .. controls (0.5,1.2) and (-1,3.8) .. (2.8,3.8);
	\end{tikzpicture}
	\caption{Relationships between various graph classes and the complexity of \textsc{Fair $k$-Division Under Conflicts} {(decision version)}. An arrow from a class $\mathcal{G}_1$ to a class $\mathcal{G}_2$ means that every graph in $\mathcal{G}_1$ is also in $\mathcal{G}_2$. Label `\textsf{\PP}' means that {for each fixed $k$} the problem is solvable in pseudo-polynomial time in the given class, and label `\textsf{\sNPc}' means that {for each fixed $k\ge 2$} the problem is strongly \NP-complete.
	Graph classes with results in the current paper are given in gray.
   All results from~\cite{ourarx} are shown with round corners and no color.
	Results depicted in rectangles follow from the inclusion of graph classes.
	Note that for the case of $\chi$-bounded graph classes, we only claim strong \NP-completeness for classes containing either the class of all bipartite graphs or the class of all line graphs of bipartite graphs.}\label{fig:Hasse-new}
\end{figure}

\subsection{Definitions and notation}
\label{sec:def}

For a positive integer $k$, we denote by $[k]$ the set $\{1,\ldots, k\}$.
All graphs considered in this paper are finite, simple, and undirected. 
A vertex in a graph $G$ is said to be \emph{isolated} if it has no neighbors. An \emph{independent set} in a graph $G$ is a set of pairwise nonadjacent vertices. 
For a graph $G = (V,E)$ and a set $X\subseteq V$, we denote by $G[X]$ the \emph{subgraph of $G$ induced by $X$}, that is, the graph with vertex set $X$ in which two vertices are adjacent if and only if they are adjacent in $G$. 

A bipartite graph $G=(A\cup B, E)$ is {\em convex} if the vertices in $A$ can be linearly ordered as $(a_1,\ldots, a_s)$ so that for every vertex $b\in B$, the neighbors of $b$ form a consecutive interval of vertices in $A$ (see, e.g.,~\cite{lipski1981efficient}).
Such an ordering of $A$ can be found in linear time using PQ-trees~\cite{bolu76}.

A bipartite graph is {\em biconvex} if both vertex sets $A$ and $B$ can be linearly ordered such that for every vertex $a \in A$ (resp.\ $b \in B$) the neighbors of $a$ (resp.\ neighbors of $b$) 
form a consecutive interval in $B$ (resp.\ in $A$).
Informally, a biconvex graph is convex on both sides.

\section{Convex bipartite graphs}
\label{sec:convex}

We will first derive our main result for the case when the conflict graph $G$ is connected. 
Later we will refer to \cite[Lemma~13]{ourarx} to show that profit profiles determined during the execution of this algorithm for every connected component (as well as for isolated vertices) can be merged together for a solution on general graphs.

It will be useful to state the following observation for the solution of {\sc Fair $k$-Division of Indivisible Items}, which can be seen as an instance of {\sc Fair $k$-Division Under Conflicts} with an edgeless conflict graph $G$.

Given an instance $I = (G,p_1,\ldots, p_k)$ of \textsc{Fair $k$-Division Under Conflicts}, a \emph{profit profile} for $I$ is a $k$-tuple $(q_1,\ldots, q_k)\in \mathbb{Z}_+^k$ such that there exists a partial $k$-coloring $(X_1,\ldots, X_k)$ of $G$ for which $q_j = p_j(X_j)$ for all $j\in [k]$.

\begin{sloppypar}
\begin{lemm}\label{thm:edgeless}
For every $k\ge 1$, the set of profit profiles for a given instance $(G,p_1,\ldots, p_k)$ of \textsc{Fair $k$-Division Under Conflicts} with \hbox{edgeless} conflict graph $G$ can be computed in time $\mathcal{O}(n(Q+1)^k)$, where $n = |V(G)|$ and $Q = \max_{1\le j\le k}p_j(V(G))$.
\end{lemm}
\end{sloppypar}

\begin{proof}
Let $V=\{v_1, \ldots, v_n\}$ and $\Pi$ be the set of all profit profiles for the given instance.
Observe that $\Pi$ is initialized by the singleton consisting only of the all-zero profit profile.
Then we iterate over the vertices of $G$ as follows.
For every $i$ from $1$ to $n$ we consider every profit profile in $\Pi$ and generate $k$ new profit profiles from it by adding a profit of $p_j(v_i)$ to the current entry at position $j$ for every $j\in[k]$.
These $k |\Pi|$ new profit profiles are added to the current set $\Pi$ (trivially ignoring duplications).
Altogether there can be at most $\mathcal{O}((Q+1)^k)$ different profit profiles with $Q = \max_{1\le j\le k}p_j(V)$.
Thus, in each of the $n$ iterations at most $\mathcal{O}((Q+1)^k)$ profit profiles have to be considered to generate $k$ (i.e., constantly many) candidates for the new profit profiles.
\end{proof}

The main theorem of this section gives a pseudo-polynomial algorithm for connected convex graphs.
{Our algorithm is based on an approach similar to that of D\'iaz, Diner, Serna, and Serra~\cite{Diaz2021}, relying on certain structural properties of convex bipartite graphs.}

\begin{theorem}\label{thm:convex}
For every $k\ge 1$, \textsc{Fair $k$-Division Under Conflicts} is solvable in time $\mathcal{O}(n^{3k+1}(Q+1)^{2k})$ for connected $n$-vertex convex bipartite conflict graphs $G$, where $Q=\max_{1\le j\le k}p_j(V(G))$.
\end{theorem}

{Before proving Theorem~\ref{thm:convex} we introduce some notation.}
Let $G = (A\cup B,E)$ be a connected convex bipartite graph, with vertices in $A$ linearly ordered as $(a_1,\ldots, a_s)$ so that for every vertex $b\in B$, the neighbors of $b$ form a consecutive interval of vertices in $A$.
Recall that such an ordering of $A$ can be found in linear time using PQ-trees~\cite{bolu76}.
Furthermore, we assume that $G$ has at least two vertices. 
In particular, since $G$ is connected, this implies that $G$ does not have any isolated vertices, and sets $A$ and $B$ are both non-empty.

For each $b\in B$, let $b^-$ and $b^+$ be the two elements of $[s]$ such that 
$N_G(b) = \{a_j\mid b^-\le j\le b^+\}$. 
Note that since $G$ has no isolated vertices, 
$b^-$ and $b^+$ are well-defined, but may coincide.
We may also assume that vertices of $B$ are sorted linearly as $(b_1,\ldots, b_t)$ so that  for all $1\le i<j\le t$, we have $b_i^+\le b_j^+$, and if  $b_i^+= b_j^+$, then $b_i^-\le b_j^-$. In other words, we consider vertices of $B$ to be ordered non-decreasingly with respect to the larger endpoints of the intervals representing their neighborhoods, and in case of a tie, they are ordered non-decreasingly with respect to their smaller endpoints.
See Fig.~\ref{fig:example} for an example.

\begin{figure}[h]
\begin{center}
\includegraphics[width=0.7\textwidth]{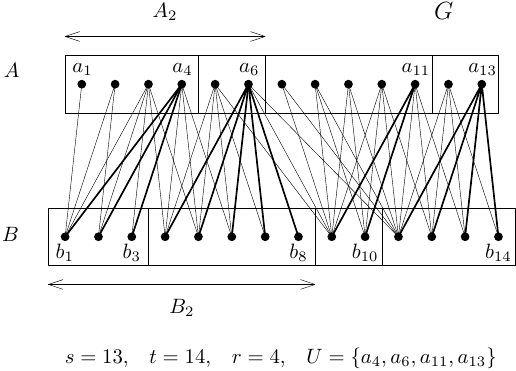}

\bigskip

\begin{tabular}{|c?{0.6mm}ccc?{0.6mm}ccccc?{0.6mm}cc?{0.6mm}cccc|}
  \hline
  $i$ & $1$ & $2$ & $3$ & $4$ & $5$ & $6$ & $7$ & $8$ & $9$ & $10$ & $11$ & $12$ & $13$ & $14$\\
  \hline
  $b_i^-$ & $1$ & $2$ & $3$ & $3$ & $3$ & $4$ & $5$ & $6$ & $5$ & $8$ & $6$ & $10$ & $11$ & $12$\\
  \hline
  $b_i^+$ & $4$ & $4$ & $4$ & $6$ & $6$ & $6$ & $6$ & $6$ & $11$ & $11$ & $13$ & $13$ & $13$ & $13$\\
  \hline
\end{tabular}

\bigskip 

\begin{tabular}{|c?{0.6mm}cccc|}
  \hline
  $j$ & $1$ & $2$ & $3$ & $4$ \\
  \hline
  $u_j$ & $4$ & $6$ & $11$ & $13$ \\
    \hline
  $v_j$ & $3$ & $8$ & $10$ & $14$ \\
  \hline
\end{tabular}
\caption{An example of a convex bipartite graph and the corresponding indices and sets as used in the proof.}\label{fig:example}
\end{center}
\end{figure}

First, we compute the set $U\subseteq A$ of the larger endpoints of the neighborhoods of vertices in $B$, that is, $U  = \{a_{b^+}\mid b\in B\}$, and sort the elements of $U$ increasingly by their index: $U = \{a_{u_1}, \ldots, a_{u_r}\}$ such that $u_1<\ldots < u_r$. 
(Note that $u_r = s$ since $G$ is connected.)
For all $j\in [r]$, we set: 
\begin{align*}
    A_j &= \{a_i\in A\mid 1\le i\le u_j\},\\
    B_j &= \{b\in B\mid b^+\le u_j\},\,\text{and}\\
    G_j &= G[A_j\cup B_j].
\end{align*}

The graphs $G_1,\ldots, G_r$ form an increasing sequence (with respect to the induced subgraph relation) of induced subgraphs of $G$ such that $G_r = G$.
For each $j\in [r]$, analogously to $u_j$, we also define $v_j\in [t]$ to
be the index such that 
 $B_j = \{b_1, \ldots, b_{v_j}\}$.
Note that $v_1<\ldots <v_r$.
We now point out some immediate but important facts regarding the above notions:
\begin{enumerate}[label=(\roman*)]
   \item $A_1\subsetneq \ldots \subsetneq A_r$, 
 $B_1\subsetneq \ldots \subsetneq B_r$.
  \item $A_j\setminus A_{j-1} = \{a_{u_{j-1}+1},\ldots, a_{u_j}\}$ and  $B_j\setminus B_{j-1} = \{b_{v_{j-1}+1},\ldots, b_{v_j}\}$. 
  \item 
  For all $i,i'\in \{u_{j-1}+1,\ldots, u_j\}$ with $i<i'$, we have $N_{G_j}(a_{i})\subseteq N_{G_j}(a_{i'})\subseteq B_j\setminus B_{j-1}$.\label{item:convex-4}
\item 
  For all $i,i'\in \{v_{j-1}+1,\ldots, v_j\}$ with $i<i'$, we have $N_{G_j}(b_{i'})\subseteq N_{G_j}(b_i)$.\label{item:convex-3}
\end{enumerate}

{After these notational preparations we proceed with the proof of Theorem \ref{thm:convex}. 
}

\begin{proof}[Proof of Theorem~\ref{thm:convex}]

{The proof of the theorem is constructive and relies on dynamic programming.
Let us start with a high-level intuitive explanation:
We develop a dynamic programming algorithm that considers a number of restricted subproblems on graphs $G_j$ for increasing~$j$.
For each graph $G_j$, $j\in \{1,\dots,r\}$ (as defined above) we \emph{guess} (by enumerating all possibilities) the largest indexed vertex in $V(G_j)\cap A = A_j$ for every color.
This means that we go through all $k$-tuples 
$(i_1,\ldots, i_k)\in \{0,1,\ldots,u_j\}^k$ and compute the set $\Pi_j(i_1,\ldots, i_k)$, defined as the set of all profit profiles (see next paragraph for the definition) of $G_j$ such that each agent $\ell\in [k]$ is assigned vertex $a_{i_\ell}$ and possibly other vertices from $\{a_1, \ldots, a_{i_\ell-1}\} \cup B_j$.
As the set of all profit profiles of partial $k$-colorings of $G$ equals the union, over all $(i_1,\ldots, i_k)\in \{0,1,\ldots, u_r\}^k$, of the sets $\Pi_r(i_1,\ldots, i_k)$, this will give us a solution for our problem.
}

\paragraph{{Definition of profit profiles $\Pi_j (i_1,\dots,i_k)$.}}
\begin{sloppypar}
For each $j\in \{1,\dots,r\}$ and each $k$-tuple $(i_1,\ldots, i_k)\in \{0,1,\dots,u_j\}^k$, the set $\Pi_j(i_1,\ldots, i_k)$ of profit profiles is defined to be the set of all $k$-tuples $(q_1,\ldots, q_k)\in \mathbb{Z}_+^k$ such that there exists a partial $k$-coloring $(X_1,\ldots, X_k)$ of $G_j$ with the following property:
for all $\ell\in [k]$, we have $q_\ell = p_\ell(X_\ell)$ and
\begin{equation}\label{eq0}
i_\ell = \left\{
  \begin{array}{ll}
    \max\{u \mid a_u\in X_\ell \}, & \hbox{if $X_\ell\cap A\neq \emptyset$;} \\
    0, & \hbox{if $X_\ell\cap A = \emptyset$}.
  \end{array}
\right.
\end{equation}
\end{sloppypar}


Note that for each $\ell\in [k]$, the possible values of the $\ell$-th coordinate of any member of $\Pi_j(i_1,\ldots, i_k)$ belong to the set $\{0,1,\ldots, Q\}$ where $Q = \max_{1\le j\le k}p_j(V(G))$. Thus, each set $\Pi_j(i_1,\ldots, i_k)$ has at most $(Q+1)^k$ elements. 
Note also that for each fixed $j\in [r]$, the total number of sets $\Pi_j(i_1,\ldots, i_k)$ is of the order
$\mathcal{O}(n^{k})$. 
The whole remainder of the proof deals with the explanation of how to compute the sets $\Pi_j(i_1,\ldots, i_k)$ for each $j\in [r]$ and each $k$-tuple $(i_1,\ldots, i_k)\in \{0,1,\ldots, u_j\}^k$. 

\paragraph{{Computation of profit profiles for $j=1$.}}
{In the initial case $G_1$ is the subgraph of $G$ induced by $A_1\cup B_1$,} where $A_1 = \{a_1,\ldots, a_{u_1}\}$ and
$B_1 = \{b_1,\ldots, b_{v_1}\}$. 
Recall from properties~\ref{item:convex-4} and \ref{item:convex-3} that 
$N_{G_1}(a_{i})\subseteq N_{G_1}(a_{i'})$ for all $i,i'\in [u_1]$ with $i<i'$,
and similarly
$N_{G_1}(b_{i'})\subseteq N_{G_1}(b_i)$ for all $i,i'\in [v_1]$ with $i<i'$.
Consider a $k$-tuple $(i_1,\ldots, i_k)\in \{0,1,\ldots, u_1\}^k$. We may assume that 
$i_{\ell} = i_{\ell'}$ if and only if  $\ell = \ell'$ or $i_{\ell} = i_{\ell'} = 0$,
since vertex $a_{i_\ell}$ can be assigned at most one color; in other words, for all $k$-tuples $(i_1,\ldots, i_k)$ not satisfying the above condition, we clearly have $\Pi_1(i_1,\ldots, i_k) = \emptyset$.
Note also that for each  $\ell\in [k]$ such that $i_\ell> 0$,
vertex $a_{i_\ell}$ must be assigned to $X_\ell$, the $\ell$-th color class,
no vertex in $A_1$ with index larger than $i_\ell$ is assigned to $X_\ell$, and, moreover,
no vertex in $B_1$ adjacent to $a_{i_\ell}$ can be assigned to $X_\ell$.
To consider possible additional assignments of vertices to $X_\ell$ we define $I_\ell$ as the set containing 
vertices in $B_1$ nonadjacent to $a_{i_\ell}$ and vertices in $A_1$ with index strictly smaller than $i_\ell$.

\medskip
\noindent\textbf{Observation.} 
$I_\ell$ is an independent set in $G_1$.

\begin{proof}
Suppose by contradiction that $a_ib_{i'}\in E(G_1)$, with $i<i_\ell$ and $a_{i_\ell}b_{i'}\notin E(G_1)$. 
By construction of $G_1$ we clearly have $a_{u_1}b_{i'}\in E(G_1)$.
But since $i<{i_\ell}\le {u_1}$, this violates the assumption that the neighbors of $b_{i'}$ form a consecutive interval of vertices in $A$.
\end{proof}

\begin{sloppypar}
Set $Y_1 = \{a_{i_\ell}\mid 1\le \ell\le k \wedge i_\ell> 0\} \subseteq A_1$.
Now notice that the remaining elements of $X_\ell$ can be any of the vertices in $I_\ell\setminus Y_1$.
The reasoning is similar for all $\ell\in [k]$ where $i_\ell= 0$ since in this case no vertex in $A_1$ is allowed to be assigned to $X_\ell$, but any vertex in $B_1$ can be.
We thus define, for all $\ell\in [k]$, 
a modified profit function $p'_\ell$ on the vertices of $G_1 \setminus Y_1$
which, informally speaking, assigns the original profits only to vertices in $I_\ell$.
Formally, we set for all $v\in V(G_1 \setminus Y_1)$:
$$p'_\ell(v) = \left\{
                 \begin{array}{ll}
                   0, & \hbox{if $i_\ell=0$ and $v\in A_1$;}\\ 
                   0, & \hbox{if $i_\ell>0$ and $v\in \{a_{i_\ell+1},\ldots, a_{u_1}\}\cup N_{G_1}(a_{i_\ell})$;} \\
                   p_\ell(v), & \hbox{otherwise.}
                 \end{array}
               \right.
$$
Note that for each $\ell$, no two vertices of $G_1\setminus Y_1$ with strictly positive value of the modified profit function $p_\ell'$ are adjacent.
Thus, in order to compute the set $\Pi_1'$ of all possible profit profiles with respect to the modified profits $(p_1',\ldots, p_k')$ on the set $V(G_1)\setminus Y_1$, we can completely ignore the edges, and hence this set can be computed using Lemma~\ref{thm:edgeless}.
Adding the profit implied by index $i_\ell$ we obtain
$$\Pi_1(i_1,\ldots, i_k) = \{(q_1,\ldots, q_k) \mid (q_1',\ldots, q_k') \in \Pi_1'\}$$
where $$q_\ell =  \left\{
                   \begin{array}{ll}
                     q_\ell'+p_\ell(a_{i_\ell}), & \hbox{if $i_\ell>0$;} \\
                     q_\ell', & \hbox{otherwise.}
                   \end{array}
                 \right.$$
This computation is carried out for each of the 
${\mathcal O}(n^k)$ $k$-tuples $(i_1,\ldots, i_k)\in \{0,1,\ldots, u_1\}^k$.
\end{sloppypar}

\paragraph{{Computation of profit profiles for $j>1$.}}
After settling the case $j=1$, we now consider $j > 1$ and assume that all the relevant sets $\Pi_{j-1}(i_1',\ldots, i_k')$ were already computed.

\begin{sloppypar}
The main loop for a fixed value of $j$ considers an arbitrary but fixed $k$-tuple \hbox{$(i_1,\ldots, i_k)\in\{0,1,\ldots, u_j\}^k$}.
Again, we may assume that $i_{\ell} = i_{\ell'}$ if and only if  $\ell = \ell'$ or $i_{\ell} = i_{\ell'} = 0$,
since otherwise we again have $\Pi_j(i_1,\ldots, i_k) = \emptyset$.
\end{sloppypar}

The main idea is a systematic structuring of the contributions from $B_j\setminus B_{j-1}$.
Therefore, we first guess the smallest indexed vertices in $B_j\setminus B_{j-1}$ that will belong to each color class $X_\ell$ (of course, this guess has to be performed for every $k$-tuple $(i_1,\ldots, i_k)$).
More formally, we consider all possible $k$-tuples $(m_1,\ldots, m_k)\in \{v_{j-1}+1,\ldots, v_j,\infty\}^k$, 
where $\infty$ is a dummy symbol such that $t<\infty$.
To prevent assigning the same vertex to two agents, we restrict ourselves to the $k$-tuples $(m_1,\ldots, m_k)$ for which $m_{\ell} = m_{\ell'}$ if and only if  $\ell = \ell'$ or $m_{\ell} = m_{\ell'} = \infty$.
We say that such a $k$-tuple $(m_1,\ldots, m_k)$ is {\it compatible with} the current $k$-tuple $(i_1,\ldots, i_k)$ from the $A$-side
 if for all $\ell\in [k]$ such that 
$i_\ell>0$ and $m_\ell<\infty$, vertices $a_{i_\ell}$ and $b_{m_\ell}$ are nonadjacent in $G_j$ (or, equivalently, in $G$).
We denote by {$\mathcal{M}_j = \mathcal{M}_j(i_1,\ldots, i_k)$} the set of all such
$k$-tuples $(m_1,\ldots, m_k)$ compatible with the current guess $(i_1,\ldots, i_k)$.
{For simplicity of notation, in the rest of the proof we do not explicitly write the dependency on the current guess $(i_1,\ldots, i_k)$.} 

Next, we define the following set:
\begin{align*}
    Y_j &= \{a_{i_\ell}\mid 1\le \ell\le k \wedge i_\ell> 0\}\setminus A_{j-1},
    \end{align*}
and, for each $k$-tuple $(m_1,\ldots, m_k)\in \mathcal{M}_j$, a set $Z_j$ (analogous to $Y_j$ on the $B$-side) and with it the resulting graph $G_j'$ representing vertices added in the current iteration $j$ and still open for assignment:
\begin{align*}
    Z_j &= \{b_{m_\ell}\mid 1\le \ell\le k \wedge m_\ell <\infty\}\,, \\
G_j' &= G[(A_{j}\setminus (A_{j-1}\cup Y_j))\cup (B_j\setminus (B_{j-1}\cup Z_j)]\,.
\end{align*}
We define, for all $\ell\in [k]$, a modified profit function $p'_\ell$ on the vertices of $G_j'$ which assigns the original profits to those vertices that still might be assigned to color $\ell$ without violating the definitions of $i_\ell$ and $m_\ell$.
We set for all $v\in V(G_j')$:
$$p'_{\ell}(v)=
    \begin{cases}
        0, & \text{if }v\in\left\{ a_{\max\{i_{\ell},u_{j-1}\}+1},\ldots,a_{u_{j}}\right\} ;\\
        0, & \text{if }i_\ell >0 \text{ and } v\in N_{G_j}(a_{i_\ell});\\
        0, & \text{if }v\in\left\{ b_{v_{j-1}+1},\ldots,b_{\min\{m_{\ell},v_{j}+1\}-1}\right\};\\
        0, & \text{if }m_\ell<\infty \text{ and } v\in N_{G_j}(b_{m_\ell});\\
        p_{\ell}(v), & \text{otherwise.}
    \end{cases}
$$
Also here, similarly to the case $j = 1$, we will need an independence property of the graph $G_j'$ that can be derived from the convexity of $G$.

\medskip
\noindent\textbf{Observation.} For each $\ell$, no two vertices of $G_j'$ with strictly positive value of the modified profit function $p_\ell'$ are adjacent.

\begin{proof}
Suppose by contradiction that $a_ib_{i'}\in E(G_j')$
where $p_\ell(a_i)>0$ and $p_\ell(b_{i'})>0$.
Then $i<i_\ell$ and $i'>m_\ell$.
Since the $k$-tuple $(m_1,\ldots, m_k)$ is compatible with $(i_1,\ldots, i_k)$, vertices $a_{i_\ell}$ and $b_{m_\ell}$ are nonadjacent in $G_j$.
Furthermore, by construction of $G_j$ vertices $a_{u_j}$ and $b_{i'}$ are adjacent in $G_j$ (and in $G$).
Now, as vertex $b_{i'}$ is adjacent in $G$ to both $a_i$ and $a_{u_j}$, the inequalities $i<{i_\ell}\le {u_j}$ and the convexity property imply that 
$b_{i'}$ is adjacent in $G$ (and thus in $G_j$) to $a_{i_\ell}$.
Finally, applying property~\ref{item:convex-3} to 
$m_\ell<i'$, we obtain that 
$N_{G_j}(b_{i'})\subseteq N_{G_j}(b_{m_\ell})$, 
which contradicts the fact that 
$a_{i_\ell}$ is adjacent in $G_j$ to $b_{i'}$ but not to 
$b_{m_\ell}$.
\end{proof}

Now we determine the set $\Pi_j'(m_1,\ldots, m_k)$ of all
possible profit profiles on the set $V(G_j')$ with respect to the modified profits $(p_1',\ldots, p_k')$.
In the computation of $\Pi_j'(m_1,\ldots, m_k)$ we can again completely ignore the edges, and hence also this set can be computed as described in Lemma~\ref{thm:edgeless}.

\paragraph{{Combining profit profiles after each increment of $j$.}}
To complete the construction of $\Pi_{j}(i_{1},\dots,i_{k})$ it remains to combine these profit profiles
for $G_j'$
with the profits implied by the choice of the two $k$-tuples on the $A$- and $B$-side (in sets $A_j\setminus A_{j-1}$ and $B_j\setminus B_{j-1}$, respectively) and with those computed recursively on the graph $G_{j-1}$.
Note that there are no edges in $G_j$ joining a vertex in $A_j\setminus A_{j-1}$ with a vertex in $B_{j-1}$.
Moreover, for every $k$-tuple $(m_1,\ldots, m_k)\in \mathcal{M}_j$, 
we already know both the largest vertex in $A_j$ (if any) and the smallest vertex in $B_j\setminus B_{j-1}$ (if any) that belongs to each $X_\ell$.

Let $\mathcal{I}_j = \mathcal{I}_j(i_1,\ldots, i_k)$ be the set of all $k$-tuples $(i_1',\ldots, i_k')\in \{0,1,\ldots, u_{j-1}\}^k$ for which $i_\ell' = i_\ell$ for all $\ell\in [k]$ such that $i_\ell\le u_{j-1}$. 
The intuition behind the definition of the set $\mathcal{I}_j$ is as follows. The corresponding $k$-tuples encode the choices for $i_\ell$ that were available already in the previous iteration for $j-1$ and therefore for each $(i_1',\ldots, i_k')\in \mathcal{I}_j$, the set $\Pi_{j-1}(i_1',\ldots, i_k')$ has already been computed. More precisely, consider a partial coloring $(X_1,\ldots, X_k)$ of $G_j$ such that $i_\ell = \max\{u\mid a_u\in X_\ell\}$ if $X_\ell\cap A \neq \emptyset$ and $i_\ell = 0$ if $X_\ell\cap A = \emptyset$. 
We consider three cases depending on the value of $i_\ell$:
\begin{itemize}
    \item If $i_\ell = 0$, then no vertex from $A_j$ is assigned to agent $\ell$ and thus the same condition will hold for the restriction of $X_\ell$ to $G_{j-1}$; thus we impose the condition $i_\ell' = 0 = i_\ell$ in this case.
\item If $0 < i_\ell \le u_{j-1}$, then we again want to require that $i_\ell' = i_\ell$ since in this case vertex $a_{i_\ell}$ will be the largest indexed vertex assigned to agent $\ell$ not only in $A_j$ but also in $A_{j-1}$.
\item Finally, if $i_\ell > u_{j-1}$, then the largest indexed vertex assigned to agent $\ell$ belongs to $A_j\setminus A_{j-1}$, which imposes no constraints on the largest indexed vertex assigned to agent $\ell$ from the set $A_{j-1}$.
\end{itemize}

The profit contribution added in the current iteration $j$ 
directly by the choice of $(i_{1},\dots,i_{k})$ will be denoted by the $k$-tuple $q^A$.
Similarly, for every choice of $\mu=(m_1,\ldots, m_k)\in \mathcal{M}_j$ the resulting profit profile is denoted by the $k$-tuple $q^B(\mu)$.
Formally, we have the following:
\begin{align*}
    q^A&=(q_1^A,\dots,q_k^A),\text{ where } q_\ell^A = \left\{\begin{array}{ll}
                 p_\ell(a_{i_\ell}), & \hbox{if $i_\ell\ge u_{j-1}+1$;} \\
                 0, & \hbox{otherwise;}
               \end{array}
             \right.\\
    q^B(\mu)&=(q_1^B,\dots,q_k^B),\text{ where }
    q_\ell^B = \left\{
               \begin{array}{ll}
                 p_\ell(b_{m_\ell}), & \hbox{if $m_\ell< \infty$;} \\
                 0, & \hbox{otherwise.}
               \end{array}
             \right.
\end{align*}
\begin{sloppypar}
Now we are ready to paste together all parts contributing to the set of profit profiles $\Pi_{j}(i_{1},\dots,i_{k})$ in iteration $j$ for the current guess $(i_{1},\dots,i_{k})$.
Namely, we take the union over all guesses from the preceding iteration $j-1$ which are compatible with $(i_{1},\dots,i_{k})$,
i.e.\ the union over all $k$-tuples $\tau=(i_1',\ldots, i_k')\in \mathcal{I}_j$.
Recall that any such $\tau$ coincides with the current guess for all indices $i_\ell \leq u_{j-1}$.
From the previous iteration we have at hand a set of potential profit profiles $\Pi_{j-1}(\tau)$ for each such $\tau$.
Moreover, we also take the union over all possible choices $\mu\in \mathcal{M}_j$, which are compatible with the current guess $(i_{1},\dots,i_{k})$ by definition of $\mathcal{M}_j$.
The possible contributions from the newly considered vertices $V(G_j\setminus G_{j-1})$ in compliance with $(i_{1},\dots,i_{k})$ and $\mu$ are added through $\Pi'_{j}(\mu)$.
Altogether this yields:
\begin{align}\label{eq:final}
    \Pi_{j}(i_{1},\dots,i_{k})=\bigcup_{\begin{array}{c}
{\scriptstyle \tau\in\mathcal{I}_{j},}\\
{\scriptstyle \mu\in\mathcal{M}_{j}}
\end{array}}\bigcup_{\begin{array}{c}
{\scriptstyle q\in \Pi_{j-1}(\tau),}\\
{\scriptstyle q'\in \Pi'_{j}(\mu)}
\end{array}}\left\{q+q'+q^{A}+q^{B}(\mu)\right\}
\end{align}
\end{sloppypar}
{Recall that, as the set of all profit profiles of partial $k$-colorings of $G$ equals the union, over all $(i_1,\ldots, i_k)\in \{0,1,\ldots, u_r\}^k$, of the sets $\Pi_r(i_1,\ldots, i_k)$,
assuming the correctness of equation (\ref{eq:final}),
this gives the solution of \textsc{Fair $k$-Division Under Conflicts} and concludes the proof.
}

\paragraph{{Correctness of equation \eqref{eq:final}.}}
First, consider a profit profile in $\Pi_j(i_1,\ldots, i_k)$. 
To see that the profile can be written as the sum $q+q'+q^A+q^B(\mu)$ for some $q\in \Pi_{j-1}(\tau)$ with $\tau\in \mathcal{I}_j$ and $q'\in \Pi'_j(\mu)$ with $\mu\in \mathcal{M}_j$, fix any partial $k$-coloring $(X_1,\ldots, X_k)$ of $G_j$ ``compatible with'' $(i_1,\ldots, i_k)$ and having the given profit profile.
For each color $\ell\in [k]$, we determine the vertex in $X_\ell\cap (B_j\setminus B_{j-1})$ (if any) with the smallest index. Such vertices define a $k$-tuple $\mu\in \mathcal{M}_j$ and define a set $Z_j$ that corresponds to a profit profile $q^B(\mu)$.
Similarly, the set of all vertices in $X_\ell\cap (A_j\setminus A_{j-1})$ with the largest index (if any), over all  $\ell\in [k]$, define a set $Y_j$ 
that corresponds to a profit profile $q^A$.
Vertices in $X_\ell\cap ((A_j\cup B_j)\setminus (Y_j\cup Z_j))$ correspond to an independent set in $G_j' = G[(A_{j}\setminus (A_{j-1}\cup Y_j))\cup (B_j\setminus (B_{j-1}\cup Z_j)]$, and all these independent sets form a partial $k$-coloring of $G_j'$ with profit profile $q'\in \Pi'_j(\mu)$.
Finally, vertices in $X_\ell\cap V(G_{j-1})$ form an independent set in $G_{j-1}$, and all these independent sets form a partial $k$-coloring of $G_{j-1}$ with profit profile $q\in \Pi_{j-1}(\tau)$. 
By construction, the profit profile of $(X_1,\ldots, X_k)$ is precisely the sum $q+q'+q^A+q^B(\mu)$.

\medskip
Conversely, consider some $q\in \Pi_{j-1}(\tau)$ with $\tau\in \mathcal{I}_j$ and $q'\in \Pi'_j(\mu)$ with $\mu = (m_1,\ldots, m_k)\in \mathcal{M}_j$.
Fix any partial $k$-coloring $(X_1,\ldots, X_k)$ of $G_{j-1}$ ``compatible with'' 
$\tau=(i_1',\ldots, i_k')$ having profit profile $q$ and set, as before, $Y_j = \{a_{i_\ell}\mid 1\le \ell\le k \wedge i_\ell> 0\}\setminus A_{j-1}$, $Z_j = \{b_{m_\ell}\mid 1\le \ell\le k \wedge m_\ell <\infty\}$, and $G_j' = G[(A_{j}\setminus (A_{j-1}\cup Y_j))\cup (B_j\setminus (B_{j-1}\cup Z_j)]$.
Fix any partial $k$-coloring $(X_1',\ldots, X_k')$ of $G_{j}'$ having a $p'$-profit profile equal to the profile $q'\in \Pi'_j(\mu)$ we considered at the beginning.
Furthermore, for all $\ell\in [k]$ and all $v\in X'_\ell$, we may assume that $p'(v)>0$ (otherwise we can simply remove $v$ from $X'_\ell$).
For each $\ell\in [k]$, let us denote by $Y^\ell$ the set $\{a_{i_\ell}\}$ if $i_\ell\in\{u_{j-1}+1\ldots, u_j\}$, and $Y^\ell = \emptyset$, otherwise.
For the considered $\mu$, we denote by $Z^\ell$ the set $\{b_{m_\ell}\}$ if $m_\ell<\infty$, and $Z^\ell = \emptyset$, otherwise.
It now suffices to show that for each $\ell\in [k]$, the set $X_\ell\cup X_\ell'\cup Y^\ell\cup Z^\ell$ is independent in $G$.
By construction, each of the sets $X_\ell$, $X_\ell'$, $Y^\ell$, and $Z^\ell$ is independent.
We justify that there are no edges between any two of these sets as follows.
\begin{itemize}
\item There are no edges in $G$ between $B_{j-1}$ and $A_j\setminus A_{j-1}$, hence there are no edges between $X_\ell\cap B$ and $Y^\ell\cup (X_\ell'\cap A_j)$. 
\item There are no edges between $X_{\ell}\cap A$ and $Z^{\ell}$.

Suppose for a contradiction that $m_\ell<\infty$ and $b_{m_\ell}$ is adjacent to a vertex $a\in X_{\ell}\cap A$.
Then $a = a_g$ for some $g\in[u_{j-1}]$ and, in particular, $i_{\ell}\ge i_{\ell}'\ge g$ and this $i_\ell>0$.
Note also that $b_{m_\ell}$ is non-adjacent to $a_{i_\ell}$, since $\mu= (m_1,\ldots, m_k)$ is compatible with $(i_1,\ldots, i_k)$.
Thus, $b_{m_\ell}$ is adjacent to $a_g$, non-adjacent to $a_{i_\ell}$, and adjacent to $a_{u_j}$.
Since $g\le i_\ell\le j$, we obtain a contradiction with the convexity of $G$.

\item There are no edges between $X_{\ell}\cap A$ and $X_\ell'\cap B$.

Suppose for a contradiction that there exists a pair of adjacent vertices $a\in X_{\ell}\cap A$ and $b\in X_\ell'\cap B$. 
Then $a = a_g$ for some $g\le [u_{j-1}]$ and $b = b_h$ for some $h\in \{v_{j-1}+1,\ldots, v_j\}$.
Again, we have $i_{\ell}\ge i_{\ell}'\ge g$ and thus $i_\ell>0$.
Since $b_h\in X_\ell'$, we have $p'(b_h)>0$ and thus $b_h$ is not adjacent to $a_{i_\ell}$.
Thus, $b_{h}$ is adjacent to $a_g$, non-adjacent to $a_{i_\ell}$, and adjacent to $a_{u_j}$.
Since $g\le i_\ell\le j$, we obtain a contradiction with the convexity of $G$.

\item There are no edges between $X_{\ell}'\cap A$ and $Z^{\ell}$.

This is true since if $m_\ell<\infty$, then all vertices $v\in X_{\ell}'\cap A$ satisfy $p'(v)>0$ and are thus non-adjacent to $b_{m_\ell}$ by the definition of the modified profit function.

\item Similarly, there are no edges between $X_{\ell}'\cap B$ and $Y^{\ell}$.

\item Finally, there are no edges between $Y^\ell$ and $Z^\ell$, since $\mu= (m_1,\ldots, m_k)$ is compatible with $(i_1,\ldots, i_k)$.
\end{itemize}
We conclude that equation~\eqref{eq:final} is correct.

\paragraph{{Time complexity analysis.}}
Concerning the running time, we restrict ourselves to a rough estimate.
The ordering of sets $A$ and $B$, along with the computation of set $U$ and indices $u_1,\ldots, u_r$ and $v_1,\ldots, v_r$ can be done in time $\mathcal{O}(|V(G)|+|E(G)|+|V(G)|\log |V(G)|)$, i.e., $\mathcal{O}(n^2)$.
Then, in each of the $\mathcal{O}(n)$ iterations over $j$ we iterate over each of the $\mathcal{O}(n^k)$ $k$-tuples $(i_{1},\dots,i_{k})\in \{0,1,\ldots, u_j\}^k$, and determine the set $\mathcal{M}_{j}$ of $\mathcal{O}(n^k)$ compatible $k$-tuples in $\mathcal{O}(n^k)$ time. (We first compute, in time $\mathcal{O}(kn)$, the index sets of non-neighbors in $B_j\setminus B_{j-1}$ of each vertex $a_{i_\ell}$ (whenever $i_\ell>0$) and then take the Cartesian product of these sets, each augmented by $\infty$.)
For every $\mu\in\mathcal{M}_{j}$ there are $\mathcal{O}((Q+1)^k)$ profit profiles in $\Pi'_{j}(\mu)$ which can be computed in
${\mathcal O}(n(Q+1)^k)$ time according to Lemma~\ref{thm:edgeless}.

The final combination of entries in (\ref{eq:final}) goes  over all
$\tau\in\mathcal{I}_{j}$, where $|\mathcal{I}_{j}|$ is in $\mathcal{O}(n^k)$,
and for every $\tau$ there are $\mathcal{O}((Q+1)^k)$ profit profiles in 
$\Pi_{j-1}(\tau)$.
Each of them is combined with each of $\mathcal{O}((Q+1)^k)$ profit profiles in $\Pi'_{j}(\mu)$ for all $\mathcal{O}(n^k)$ choices of $\mu$.
Thus, there are $\mathcal{O}(n^{2k}(Q+1)^{2k})$ candidates to be considered in (\ref{eq:final}) for inclusion in $\Pi_{j}(i_{1},\dots,i_{k})$.

This yields an overall running time of $\mathcal{O}(n^{3k+1}(Q+1)^{2k})$.
\end{proof}

If $G$ is not connected, then we compute for every component the set of all profit profiles for \textsc{Fair $k$-Division Under Conflicts} as described above and then aggregate these outputs.

Summarizing, we obtain the following generalization of Theorem~\ref{thm:convex}.

\begin{sloppypar}
\begin{coro}
\label{thm:convexdisc}
For every $k\ge 1$, \textsc{Fair $k$-Division Under Conflicts} is solvable in time \hbox{$\mathcal{O}(n^{3k+1}(Q+1)^{2k})$} for $n$-vertex convex bipartite conflict graphs $G$, where $Q = \max_{1\le j\le k}p_j(V(G))$.
\end{coro}
\end{sloppypar}

\begin{proof}

Let $G$ consist of $c$ connected components $G_i$, $i=1,\ldots,c$, with $|V(G_i)|=n_i$ and $\sum_{i=1}^c n_i=n$.
The running time is composed of two parts.
At first, the profit profile is computed for each of the $c$ components. 
The asymptotic running time can be bounded by
$$\sum_{i=1}^c n_i^{3k+1}(Q+1)^{2k} \leq
\left(\sum_{i=1}^c n_i\right)^{3k+1}\!\!(Q+1)^{2k}
=n^{3k+1}(Q+1)^{2k}.$$
Secondly, the $c$ profit profiles each of size  $\mathcal{O}((Q+1)^{k})$ have to be merged.
We start with the profit profiles determined for the first component and merge the profit profiles from the second component by considering every pair of profiles from first and second component and performing a vector addition to obtain a new profit profile.
Continuing this process, in every iteration the profit profiles of the next component are merged with the previously existing profiles.
Throughout this process the total number of profit profiles remains pseudo-polynomially bounded.
Finally, the best objective function value is determined by evaluating all profit profiles.
A more formal description of this approach was given in Lemma~13 of~\cite{ourarx}.
Altogether this second part requires $\mathcal{O}((c-1)(Q+1)^{2k})$ time, which is dominated by the effort of the first part.
\end{proof}

\section{Graphs of bounded clique-width}
\label{sec:bcw}

In this section we present a pseudo-polynomial time dynamic programming algorithm for \textsc{Fair $k$-Division Under Conflicts} for conflict graphs of {\em bounded clique-width}.
This is an improvement over the previous result for graphs of bounded treewidth, which was so far the only positive result for non-perfect graphs.

Clique-width, introduced 1993 in \cite{courengel93}, is a parameter defined by a construction process where only a limited number of vertex labels are available.
Vertices with the same label at some point must be treated uniformly in subsequent steps (see below). 
The clique-width $\textrm{cw}(G)$ of a graph $G$ is the minimum number of labels that suffice to construct $G$ in this way.
\NP-completeness and inapproximability of the clique-width of a graph were shown in \cite{felrosa09}.
For graphs of bounded clique-width many hard optimization problems admit polynomial-time algorithms, see, e.g., \cite{MR1739644,MR1905626,MR1973174,MR2323400}.

Relations between treewidth and clique-width were elaborated in
\cite{MR1743732}.
In particular, bounded treewidth $\textrm{tw}(G)$ of a graph $G$ implies bounded clique-width since 
$\textrm{cw}(G) \leq 3\cdot 2^{\textrm{tw}(G)}-1$ as shown by \cite{corneil05}.
However, the opposite implication is not true as can be seen from the family of complete graphs which have clique-width $2$ but treewidth $|V|-1$.

Another parameter of a graph $G$ related to treewidth is {\em rank-width} $\textrm{rw}(G)$ introduced in~\cite{MR2232389}.
Rank-width is also derived from a hierarchical decomposition of the graph.
Informally speaking, treewidth measures the width of a separation into two sides, whereas rank-width measures the rank of the adjacency matrix of the edges between the two sides of the separation. 
Without going into more details, let us just mention that it was shown in \cite{MR2232389} that $\textrm{rw}(G) \leq \textrm{cw}(G) \leq 2^{\textrm{rw}(G)+1}-1$.
Therefore, bounded clique-width is equivalent to bounded rank-width.

\smallskip
In the following we will describe the labelling process of the graph decomposition associated to clique-width in more detail.

A \emph{labeled graph} is a graph in which every vertex is assigned some label from $\mathbb{N}$. If all vertex labels belong to the set $[k]$, then we say that the graph is $k$-labeled. The \emph{clique-width} of a graph $G$ is defined as the smallest positive integer $k$ such that a $k$-labeled graph isomorphic to $G$ 
can be constructed with the following operations:
\begin{itemize}
\item $i(v)$: creating a new one-vertex graph with vertex $v$ labeled $i$,
\item $G\oplus H$: disjoint union of two already constructed labeled graphs $G$ and $H$,
\item $\eta_{i,j}$, for $i\neq j$: adding to $G$ all edges between vertices labeled $i$ and vertices labeled $j$,
\item $\rho_{i\to j}$, for $i\neq j$: relabeling every vertex labeled $i$ with label $j$.
\end{itemize}
A construction of a graph $G$ with the above four operations can be represented by  an algebraic expression, which is called a \emph{$k$-expression} if it
uses at most $k$ labels. 
Given a $k$-expression $\sigma$, we denote by $|\sigma|$ its encoding length.
A graph class $\mathcal{G}$ is said to be of \emph{bounded clique-width} if there exists a nonnegative integer $k$ such that each graph in $\mathcal{G}$ has clique-width at most $k$.

Polynomial-time algorithms for graphs with bounded clique-width are typically developed using dynamic programming based on a $k$-expression building the input graph.
If a $k$-expression is not available, then one can use any of the available algorithms in the literature for computing an expression with at most $f(k)$ labels for some exponential function $f$ (see~\cite{MR4490048,MR4334540,MR2479181,MR2232389}).
The currently fastest such algorithm is due to Fomin and Korhonen~\cite{MR4490048}; for an integer $k$ and an $n$-vertex graph $G$, it runs in time $2^{2^{O(k)}}n^2$ and either computes a $(2^{2k+1}-1)$-expression of $G$ or correctly determines that the clique-width of $G$ is more than $k$.

We can now proceed to prove the following theorem.

\begin{sloppypar}
\begin{theorem}\label{thm:boundedcliquewidth}
For every two positive integers $k$ and $\ell$, \textsc{Fair $k$-Division Under Conflicts} is solvable in time ${\mathcal O}\left(4^{k\ell}|\sigma|(Q+1)^{2k}\right)$, if the conflict graph $G$ has clique-width at most $\ell$ and is given by an \hbox{$\ell$-expression} $\sigma$, where $Q = \max_{1\le j\le k}p_j(V(G))$.
\end{theorem}
\end{sloppypar}

\begin{proof}
We extend the standard dynamic programming algorithm for graphs of bounded clique-width for the case $k = 1$, that is, the maximum weight independent set problem (see, e.g.,~\cite{Gurski}).
Given a partial $k$-coloring $c = (X_1,\ldots, X_k)$ of an $\ell$-labeled graph $H$, the \emph{label profile} of $c$ (with respect to $H$) is the $k$-tuple $(L_1,\ldots, L_k)$ where $L_j$ is the set of labels in $[\ell]$ appearing on some vertex of $X_j$, for all $j \in [k]$. 
For each labeled subgraph $H$ of $G$ that appears in the process of constructing $G$ using $\sigma$ and each $k$ label sets $L_1,\ldots, L_k\subseteq [\ell]$, we compute the set $P(H,L_1,\ldots, L_k)$ of all profit profiles $(q_1,\ldots, q_k)$ of partial $k$-colorings $c$ of $H$ such that the label profile of $c$ equals $(L_1,\ldots, L_k)$. 
We then have four cases depending on the type of $H$. 
In each case, we derive a formula of how to compute the set $P(H,L_1,\ldots, L_k)$ from the previously computed sets of this type.
         
\begin{enumerate}
    \item \textbf{$H$ is a one-vertex graph consisting of a vertex $v$ labeled $i$.}
   
   There are only $k+1$ partial $k$-colorings of $H$: the trivial partial $k$-coloring $\emptyset^k$ consisting of $k$ empty sets, and, for each $j\in [k]$, the partial $k$-coloring $c_j = (X_1,\ldots, X_k)$ where $X_j = \{v\}$ and $X_{j'} = \emptyset$ for all $j'\in [k]\setminus\{j\}$.
   The label profile of $\emptyset^k$ is $\emptyset^k$.
   For each $j\in [k]$, the label profile of $c_j$ is the $k$-tuple $(L_1,\ldots, L_k)$ where $L_j=\{i\}$  and $L_{j'}=\emptyset$ for all $j'\neq j$.
   Thus, denoting by $\textbf{e}_j(p_j(v))$ the $k$-tuple in $\mathbb{Z}_+^k$ with $j$-th coordinate equal to $p_j(v)$ and all the other coordinates equal to $0$, we have the following formula: 
$$P(H,L_1,\ldots, L_k)=\left\{
\begin{array}{ll}
\{\textbf{e}_j(p_j(v))\},& \textrm{if $L_j=\{i\}$ and $L_{j'}=\emptyset$ for all ${j'}\neq j$\,,}\\
\{(0,\dots, 0)\}, & \textrm{if $L_j=\emptyset$ for all ${j}\in [k]$\,,}\\
\emptyset,& {\rm otherwise}.
\end{array}    
\right.
$$

\smallskip
While in the remaining three cases, the assumptions on $H$ are different, we always describe how to compute the set $P(H,L_1,\ldots, L_k)$ for an arbitrary but fixed collection of $k$ label sets $L_1,\ldots, L_k\subseteq [\ell]$.

    \item \textbf{$H$ is the disjoint union of two labeled graphs $H_1$ and $H_2$.}
    
  Let $c=(X_1,\ldots, X_k)$ be a partial $k$-coloring of $H$ with label profile $(L_1,\ldots, L_k)$. Then for $i\in\{1,2\}$ we have that $c_i=(X_1\cap V(H_i),\ldots, X_k\cap V(H_i))$ is a partial $k$-coloring of $H_i$. Let us denote by $(L_1',\dots, L_k')$ and $(L_1'',\dots, L_k'')$ the label profiles of $c_1$ and $c_2$, respectively. Then $L_j = L_j'\cup L_j''$ for all $j\in[k]$. Furthermore, the converse direction holds as well: for any two partial $k$-colorings $c_1=(X_1',\ldots, X_k')$ and $c_2=(X_1'',\ldots, X_k'')$ of $H_1$ and $H_2$, respectively, 
     the $k$-tuple $c=(X_1'\cup X_1'', \dots, X_k'\cup X_k'') $ is a partial $k$-coloring of $H$ with label profile $(L_1'\cup L_1'',\dots, L_k'\cup L_k'')$, where $(L_1',\dots, L_k')$ and $(L_1'',\dots, L_k'')$ are the label profiles of $c_1$ and $c_2$, respectively. 
     This bijective correspondence yields the following formula:
   $$P(H,L_1,\ldots, L_k)=\bigcup_{}\{\mathbf{q_1}+\mathbf{q_2}\mid \mathbf{q_1}\in P(H_1,L_1',\dots, L_k'), \mathbf{q_2}\in P(H_2,L_1'', \dots, L_k'')\}\,,$$
   where the union is taken over all collections $(L_1',\dots, L_k')$ and $(L_1'',\dots, L_k'')$ of label sets such that $L_j'\cup L_j''=L_j$ for all $j\in [k]$.
      
 \item \textbf{$H$ is obtained from a labeled graph $H'$ by adding all edges between vertices labeled $i$ and vertices labeled $j$ where $i \neq j$.}

Assume first that there exists some $s\in [k]$ such that $\{i,j\}\subseteq L_s$ and let $c=(X_1,\dots, X_k)$ be a partial $k$-coloring of $H$ with label profile $(L_1,\dots, L_k)$. Since $\{i,j\}\subseteq L_s$, there are vertices $v_1$ and $v_2$ of $H$ labeled $i$ and $j$, respectively, such that $\{v_1,v_2\}\subseteq X_s$. 
By the assumption on $H$ all vertices labeled $i$ are adjacent in $H$ to all vertices labeled $j$, so it is not possible that $\{v_1,v_2\}\subseteq X_s$, since $X_s$ is an independent set in $H$; a contradiction. 
It follows that there is no partial $k$-coloring of $H$ with label profile $(L_1,\dots, L_k)$, so in this case $P(H,L_1,\ldots, L_k)=\emptyset.$

Assume now that for every $s\in[k]$ we have that $|L_s\cap \{i,j\}|\le 1$.
In this case, every partial $k$-coloring of $H$ with label profile $(L_1,\dots, L_k)$ is also a partial $k$-coloring of $H'$ with the same label profile (with respect to $H'$), and vice versa.
It follows that $P(H,L_1,\ldots, L_k)=P(H', L_1,\dots, L_k)$.
 
Altogether, we have the following equality:
$$P(H,L_1,\ldots, L_k)=\left\{
\begin{array}{ll}
P(H', L_1,\dots, L_k),& \hbox{if $|\{i,j\}\cap L_s|\le 1$ for all
$s\in [k]$},\\
\emptyset,& {\rm otherwise}.
\end{array}    
\right.
$$

\item \textbf{$H$ is obtained from a labeled graph $H'$ by relabeling all vertices labeled $i$ to vertices labeled $j$.}

Let $c=(X_1,\dots, X_k)$ be a partial $k$-coloring of $H$ with label profile $(L_1,\dots, L_k)$. 
Observe that it follows from the assumption on $H$ that no vertex in $H$ has label $i$. 

If there exists some $s\in[k]$ such that $i\in L_s$, then there is a vertex $v\in X_s$ labelled $i$; a~contradiction. It follows that there is no partial $k$-coloring of $H$ with label profile $(L_1,\dots, L_k)$, and we have that $P(H, L_1,\dots, L_k)=\emptyset$ in this case. 

Assume now that for all $s\in[k]$ we have that $i\notin L_s$. 
Let $I=\{s\in[k]\mid j\in L_s\}$ and consider an arbitrary $s\in I$. 
The vertices in $X_s$ form an independent set in $H$ and thus also in $H'$. 
Since $j\in I$, the set $L_s$ contains $j$ and therefore there exists a vertex $v\in X_s$ such that the label of $v$ in $H$ is $j$.
Thus, the set $X_{sj}$ of all vertices in $X_s$ labeled $j$ in $H$ is nonempty.
Furthermore, since the label in $H'$ of any vertex in $X_{sj}$ is either $i$ or $j$, the label set of $X_s$ in $H'$ depends on whether there exists a vertex in $X_{sj}$ labeled $i$ in $H'$ and whether there exists a vertex in $X_{sj}$ labeled $j$ in $H'$. 
More precisely, the dependency is as follows.
\begin{itemize}
\item If there exists a vertex in $X_{sj}$ labeled $i$ in $H'$ as well as one labeled $j$ in $H'$, then the label set of $X_s$ in $H'$ is $L_s\cup \{i\}$ (recall that $j\in L_s$).
\item If all vertices in $X_{sj}$ are labeled $i$ in $H'$, then the label set of $X_s$ in $H'$ is $(L_s\setminus \{j\})\cup \{i\}$.
\item If all vertices in $X_{sj}$ are labeled $j$ in $H'$, then the label set of $X_s$ in $H'$ is $L_s$.
\end{itemize}
We conclude that the vertices of $X_s$ form in $H'$ an independent set with label set being equal either to $L_s$, to $(L_s\setminus \{j\})\cup \{i\}$, or to $L_s\cup \{i\}$.
Therefore, $c$ is a partial coloring of $H'$ with label profile $(L_1',\dots, L_k')$ such that for all $s\in I$ we have $L_s'\in\{L_s, (L_s\setminus \{j\})\cup \{i\}, L_s\cup \{i\}\}$, and for all $s\in [k]\setminus I$ we have $L_s'=L_s$. 
Conversely, for any $k$ label sets $L_1',\dots, L_k'\subseteq [\ell]$ such that $L_s'\in\{L_s, (L_s\setminus \{j\})\cup \{i\}, L_s\cup \{i\}\}$ for all $s\in I$ and $L_s'=L_s$ for all $s\in [k]\setminus I$, any partial coloring of $H'$ with label profile $(L_1',\dots, L_k')$ is a partial coloring of $H$ with label profile $(L_1,\ldots, L_k)$. 

Altogether, we thus obtain the following equality:
$$P(H,L_1,\ldots, L_k)=
 \left\{
\begin{array}{ll}
   \emptyset,  & \hbox{if }  i\in L_s \hbox{ for some } s\in[k] \\
   \bigcup P(H', L_1',\dots, L_k'),  & \hbox{otherwise},
\end{array}
\right. 
$$
where the union in the second case is taken over all $k$-tuples of label sets $L_1',\dots, L_k'\subseteq [\ell]$ such that for all $s\in I$ we have $L_s'\in\{L_s, (L_s\setminus \{j\})\cup \{i\}, L_s\cup \{i\}\}$, and for all $s\in [k]\setminus I$ we have $L_s'=L_s$. 
\end{enumerate}
\paragraph{Time complexity analysis.}
From the $\ell$-expression $\sigma$ we compute in time $|\sigma|$ a rooted tree $T$ describing the construction of $G$. Each node of $T$ corresponds to a labeled subgraph $H$ of $G$. 
For each such subgraph $H$ we consider all the $2^{\ell k}$ different collections of $k$ label sets $(L_1,\dots, L_k)$, obtained by choosing a subset of $[\ell]$ for each coordinate.
We explain the time complexity separately for Case $2$. 
For Case $2$, we can initialize all the sets $P(H,L_1, \dots, L_k)$ to be empty and iterate
over all $4^{k\ell}$ pairs of collections $(L_1',\dots, L_k')$ and $(L_1'',\dots, L_k'')$ of label sets of $H_1$ and $H_2$. 
For each such iteration we add to $P(H,L_1, \dots, L_k)$, where $L_j = L_j'\cup L_j''$ for all $j\in [k]$, the elements of the set $\{\mathbf{q_1}+\mathbf{q_2}\mid \mathbf{q_1}\in P(H_1,L_1',\dots, L_k'), \mathbf{q_2}\in P(H_2,L_1'', \dots, L_k'')\}$ in time $(Q+1)^{2k}$.
Hence, the overall time complexity for a graph $H$ in Case $2$ is ${\mathcal O}(4^{k\ell}(Q+1)^{2k})$.
For the remaining three cases, we estimate the running time separately for each set $P(H,L_1,\ldots, L_k)$.
The expressions in the formulas for computing $P(H,L_1,\ldots, L_k)$ can be evaluated in time $\mathcal{O}(k\ell)$ in Case $1$, in time ${\mathcal O}(k\ell+(Q+1)^{k})$ in Case $3$, and in time ${\mathcal O}(k\ell\cdot 3^k\cdot (Q+1)^{k})$ in Case $4$.
Since Case $4$ dominates the other two cases, the overall time complexity for a graph $H$ resulting from Cases $1$, $3$, and $4$ is given by ${\mathcal O}(2^{k\ell} \cdot k\ell\cdot 3^k(Q+1)^{k})$.
Since $k\ell\le 2^{k\ell}$ and $3^k\le (Q+1)^{k}$ for all $Q\ge 2$, the overall time complexity of Cases $1$, $3$, and $4$ is dominated by the effort for Case $2$, which yields the claimed running time bound.
(The special case $Q=1$ would imply that for each agent, only one item has a non-zero profit. 
This could be solved trivially in time $\mathcal{O}(k)$.)
\end{proof}

We conclude this section with some remarks about another, more general solution approach to \textsc{Fair $k$-Division Under Conflicts} for graphs of bounded clique-width.
The unweighted version of \textsc{Fair $k$-Division Under Conflicts} (in its decision version) takes a graph $G$ and an integer $q$ as input and asks about the existence of a partial $k$-coloring in $G$ in which all the color classes have cardinality at least $q$.
The existence of a pseudopolynomial-time algorithm for this problem on graphs with bounded clique-width follows from a metatheorem of Courcelle and Durand~\cite[Theorem 27]{MR3459604}, and it is plausible that with a suitable adaptation of their approach, a solution for the general \textsc{Fair $k$-Division Under Conflicts} problem might also be developed.
However, as the algorithms constructed in~\cite{MR3459604} are very general, their running times are not specified precisely.
In contrast, our algorithm given in the proof of Theorem~\ref{thm:boundedcliquewidth} is directly tailored for \textsc{Fair $k$-Division Under Conflicts} and it is not difficult to analyze its running time.

\section{Graphs of bounded tree-independence number}
\label{sec:tin}

A \emph{tree decomposition} of a graph $G$ is a pair
$(T, \{B_t\}_{t\in V(T)})$ consisting of a tree $T$ and an assignment of a set $B_t\subseteq V(G)$ called a \emph{bag} to each node of $T$ such that (i) every vertex of $G$ is contained in a bag, (ii) for every edge $\{u,v\}\in E(G)$ there exists a bag containing both $u$ and $v$, and (iii) for every vertex $u\in V(G)$ the subgraph of $T$ induced by the set $\{t\in V(T)\,:\,u\in B_t\}$ is connected.
In our previous work~\cite{ourarx}, pseudopolynomial  algorithms for {\sc Fair $k$-Division Under Conflicts} were developed for the case when the conflict graph is either chordal or has bounded treewidth. 
Both algorithms are based on a dynamic programming approach along a tree decomposition of the conflict graph.
In the case of chordal graphs, each bag is a clique, while 
in the case of bounded treewidth, bags of the tree decomposition have bounded cardinality.
In both cases, the approach works because any optimal solution, that is, a partial $k$-coloring $(X_1,\ldots, X_k)$ of $G$, intersects each bag only in a bounded number of vertices.

This key property allows in fact for a generalization of both algorithms in a unified way, by requiring the conflict graph $G$ to be equipped with a tree decomposition with bounded \emph{independence number}, that is, the maximum size of an independent set of $G$ contained in a bag.
The minimum independence number of a tree decomposition of a graph $G$ is called the \emph{tree-independence number} of $G$ and denoted by $\tin(G)$.
This parameter was introduced independently by Yolov~\cite{MR3775804} and by Dallard, Milani{\v{c}}, and {\v{S}}torgel~\cite{dallard2022firstpaper}.
Graph classes with bounded treewidth have bounded tree-independence number, but bounded tree-independence number also holds for several classes of dense graphs such as chordal graphs, circular-arc graphs, or more generally \emph{$H$-graphs}, that is, intersection graphs of connected subgraphs of a subdivision of a fixed multigraph $H$ (see~\cite{MR1172354,MR4249058,MR4332111,MR4141534}) or intersection graphs of connected subgraphs of a graph with bounded treewidth~\cite{MR1642971}, classes of graphs in which all minimal separators are of bounded size~\cite{MR1852483}, and classes of graphs excluding a single fixed graph $H$ as an induced minor, where $H$ is either $K_5^-$, $W_4$, or a complete bipartite graph $K_{2,n}$ for some $n$~\cite{dallard2022secondpaper}.
Furthermore, several recent works (see~\cite{MR4187151,MR4276552,MR4456180}) use tree decompositions with special properties of bags that imply bounded tree-independence number.

Yolov~\cite{MR3775804}, Dallard, Milani{\v{c}}, and {\v{S}}torgel~\cite{dallard2022firstpaper}, and Milani{\v{c}} and  Rz\c{a}\.{z}ewski~\cite{milanicRzazewski2022} identified several algorithmic problems that can be solved in polynomial time if the input graph is equipped with a tree decomposition with bounded independence number.
These include the Maximum Weight Independent Packing problem, which is a common generalization of the Independent Set and Induced Matching problems, and the problem of finding a large induced sparse subgraph satisfying an arbitrary but fixed property expressible in counting monadic second-order logic.
As shown by Dallard et al.~\cite{dallard2022computing}, for every constant $k \ge 4$ it is \NP-complete to decide if a given graph has tree-independence number at most $k$.
However, for algorithmic applications of bounded tree-independence number, this is not necessarily an obstacle, since Dallard et al.~also gave an algorithm that, given an $n$-vertex graph $G$ and an integer $k$, in time $2^{\mathcal{O}(k^2)}n^{\mathcal{O}(k)}$ either outputs a tree decomposition of $G$ with independence number at most $8k$, or determines that the tree-independence number of $G$ is larger than~$k$.
 
The approach developed in our previous work~\cite{ourarx} leading to dynamic programming algorithms for {\sc Fair $k$-Division Under Conflicts} in chordal graphs and graphs with bounded treewidth can also be used to prove the following more general result.  
 
\begin{theorem}\label{thm:boundedtree-alpha}
For every two positive integers $k$ and $\ell$, \textsc{Fair $k$-Division Under Conflicts} is solvable in time $\mathcal{O}(bn^{k\ell+1}(Q+1)^{2k})$ if the input $n$-vertex conflict graph $G$ is equipped with a tree decomposition with $b$ bags and independence number at most $\ell$, where $Q = \max_{1\le j\le k}p_j(V(G))$.
\end{theorem}

\begin{proof}
The approach is a straightforward adaptation of the approach used in Sections 3.3 and 3.4 of~\cite{ourarx}, hence we only explain the main ideas.
Given a tree decomposition of $G$ with $b$ bags and independence number at most $\ell$, we first invoke~\cite[Lemma 5.1]{dallard2022firstpaper} to compute in time $\mathcal{O}(n^2b)$ a nice rooted tree decomposition $\mathcal{T} = (T, \{B_t\}_{t\in V(T)})$ of $G$ with $\mathcal{O}(nb)$ bags and independence number at most~$\ell$ (we refer to~\cite{MR3380745} for the definition of a nice tree decomposition).
For a node $t$ of $T$, let $V_t$ denote the union of all bags $B_{t'}$ such that $t'\in V(T)$ is a (not necessarily proper) descendant of $t$ in $T$.
In our dynamic programming approach we then traverse the tree $T$ bottom-up and compute, for every node $t\in V(T)$ and every partial $k$-coloring $c$ of the subgraph of $G$ induced by $B_t$, the family $P(t,c)$ of all profit profiles of partial $k$-colorings of the graph $G[V_t]$ that agree with $c$ on $B_t$.
The maximum satisfaction level over all profit profiles in the set $P(c,t)$ where $t$ is the root of $T$ will give the optimal value of \textsc{Fair $k$-Division Under Conflicts} for ($G,p_1,\ldots, p_k)$.

\paragraph{Time complexity analysis.}
Since each bag induces a subgraph with independence number at most $\ell$, each partial $k$-coloring $c = (X_1,\ldots, X_k)$ of the subgraph of $G$ induced by $B_t$ satisfies $|X_j|\le \ell$ for all $j\in [\ell]$.
It follows that the number of such partial $k$-colorings is bounded by $\mathcal{O}(n^{k\ell})$.
For a given node $t$ of $T$ and a partial $k$-coloring $c$ of $G[B_t]$, the family $P(t,c)$ of profit profiles of partial $k$-colorings of the graph $G[V_t]$ that agree with $c$ on $B_t$ can be computed from the analogous families already computed at the children of $t$ (if any).
This is done by using exactly the same recurrence relations as the ones used for graphs of bounded treewidth, see Section 3.4 of~\cite{ourarx}.
Indeed, the proof of correctness of those relations does not depend on the structure of the subgraphs induced by the bags, the only difference resulting from replacing bounded treewidth with bounded tree-independence number is in the running time.
For a given node $t$ of $T$ and a partial $k$-coloring $c$ of $G[B_t]$, the set $P(t,c)$ can be computed in time  $\mathcal{O}((Q+1)^{2k})$.
Since $T$ has $\mathcal{O}(b n)$ nodes and for each node $t$ there are $\mathcal{O}(n^{k\ell})$ partial $k$-colorings $c$ of the graph $G[B_t]$, the total number of families $P(t,c)$ to compute is of the order $\mathcal{O}(b n^{k\ell+1})$.
The claimed time complexity of the algorithm follows. 
\end{proof}

The case of chordal graphs corresponds to $\ell = 1$ and $b =\mathcal{O}(n)$ (see, e.g.,~\cite{MR1971502}) and indeed, for these values of $\ell$ and $p$, Theorem~\ref{thm:boundedtree-alpha} implies~\cite[Theorem 19]{ourarx}.
For graphs of bounded treewidth Theorem~\ref{thm:boundedtree-alpha} also implies a pseudopolynomial algorithm, although its running time does not match the corresponding tailor-made result of~\cite[Theorem 20]{ourarx}.

\section{Conclusions}\label{sec:conc}

	In this paper we continued the study of the computational complexity of \textsc{Fair $k$-Division Under Conflicts} introduced in our previous paper \cite{ourarx}.
	We managed to contribute three new dynamic programming algorithms running in pseudo-polynomial time and allowing the construction of an FPTAS.
 First, we answered a question posed in \cite{ourarx} by giving an algorithm for convex bipartite conflict graphs.
	This extends the earlier result for biconvex bipartite graphs, although it employs a totally different algorithmic strategy.
 Second, we gave algorithms for conflict graphs of bounded clique-width or bounded tree-independence number. 
	These results replace the previously derived algorithm for graphs of bounded treewidth as the currently most general positive results for non-perfect conflict graphs.
	Note that all of these dynamic programming algorithms also permit the construction of fully polynomial time approximation schemes (FPTAS).

It would be interesting to identify further graph width parameters leading to pseudopolynomial algorithms for \textsc{Fair $k$-Division Under Conflicts} and similar problems. 
{One such parameter, kindly brought to our attention by an anonymous reviewer, is thinness (see~\cite{MR2294675}),
for which the framework of Bonomo and de Estrada~\cite{MR3958231} can be adapted to problems such as \textsc{Fair $k$-Division Under Conflicts}, leading to an alternative pseudo-polynomial-time algorithm for the class of convex bipartite graphs, and even for the more general class of interval bigraphs~\cite{MR4433312}, which have thinness at most $2$ (see~\cite{MR4606115}).
This answers another question posed in \cite{ourarx}.

A natural width parameter, more general than thinness and possibly also leading to pseudopolynomial algorithms for \textsc{Fair $k$-Division Under Conflicts}, is mim-width (see~\cite{MR4087200}), as suggested to us by Andrea Munaro (personal communication, 2020).
An algorithmic metatheorem for this parameter was recently developed by Bergougnoux, Dreier, and Jaffke~\cite{MR4538078}; however, the theorem does not capture max-min problems such as \textsc{Fair $k$-Division Under Conflicts} and it is an interesting question whether an appropriate adaptation is possible.

 Finally, since \textsc{Fair $k$-Division Under Conflicts} is $\textrm{NP}$-hard on general bipartite conflict graphs (see~\cite{ourarx}), using the following chain of inclusions: convex bipartite $\subseteq$
	interval bigraph $\subseteq$ 
	chordal bipartite $\subseteq$ bipartite, we see that the complexity of \textsc{Fair $k$-Division Under Conflicts} in the class of chordal bipartite graphs remains an interesting open problem for future research.
 }

\subsubsection*{Acknowledgements.}

The authors are grateful to \"Oznur Ya\c{s}ar Diner for sharing the slides of her talk at CTW 2020 (see~\cite{Diaz2021}), which inspired our approach for convex bipartite graphs, to Mamadou M.~Kant\'e for helpful discussions and for bringing to their attention the work~\cite{MR3459604}, to Benjamin Bergougnoux, Lars Jaffke, and Andrea Munaro for insightful discussions on mim-width, and to Nevena Piva\v{c} for initial discussions on the topics of the paper. 
The authors are also grateful to the anonymous reviewers for their helpful comments.

The work of this paper was done in the framework of two bilateral projects between University of Graz and University of Primorska, financed by the OeAD (
   SI 31/2020 and SI 13/2023) and the Slovenian Research Agency (
   BI-AT/20-21-015 and BI-AT/23-24-009).
	The authors acknowledge partial support of the Slovenian Research Agency (I0-0035, research programs P1-0285, P1-0383, and P1-0404, research projects N1-0102, N1-0160, N1-0210, J1-3001, J1-3002, J1-3003, J1-4008, J1-4084, and J5-4596) and 
 the Republic of Slovenia (Investment funding of the Republic of Slovenia and the European Union of the European Regional Development Fund)
 and by the Field of Excellence ``COLIBRI'' at the University of Graz and by the Federal Ministry for Digital and Economic Affairs of the Republic of Austria through the COIN project FIT4BA.


\newpage

\bibliographystyle{abbrv}
\bibliography{Fair_allocation_structured_conflicts}

\begin{thebibliography}{10}

\bibitem{MR2277128}
N.~Bansal and M.~Sviridenko.
\newblock The {S}anta {C}laus problem.
\newblock In {\em S{TOC}'06: {P}roceedings of the 38th {A}nnual {ACM}
  {S}ymposium on {T}heory of {C}omputing}, pages 31--40. ACM, New York, 2006.

\bibitem{MR989117}
C.~Berge.
\newblock Minimax relations for the partial {$q$}-colorings of a graph.
\newblock {\em Discrete Mathematics}, 74(1-2):3--14, 1989.

\bibitem{MR4538078}
B.~Bergougnoux, J.~Dreier, and L.~Jaffke.
\newblock A logic-based algorithmic meta-theorem for mim-width.
\newblock In {\em Proceedings of the 2023 {A}nnual {ACM}-{SIAM} {S}ymposium on
  {D}iscrete {A}lgorithms ({SODA})}, pages 3282--3304. SIAM, Philadelphia, PA,
  2023.

\bibitem{MR1172354}
M.~B\'{\i}r\'{o}, M.~Hujter, and Z.~Tuza.
\newblock Precoloring extension. {I}. {I}nterval graphs.
\newblock {\em Discrete Mathematics}, 100(1-3):267--279, 1992.

\bibitem{MR1642971}
H.~Bodlaender, J.~Gustedt, and J.~A. Telle.
\newblock Linear-time register allocation for a fixed number of registers.
\newblock In {\em Proceedings of the {N}inth {A}nnual {ACM}-{SIAM} {S}ymposium
  on {D}iscrete {A}lgorithms}, pages 574--583. ACM, New York, 1998.

\bibitem{boja95}
H.~Bodlaender and K.~Jansen.
\newblock On the complexity of scheduling incompatible jobs with unit-times.
\newblock In {\em MFCS '93: Proceedings of the 18th International Symposium on
  Mathematical Foundations of Computer Science}, pages 291--300. Springer,
  1993.

\bibitem{MR3958231}
F.~Bonomo and D.~de~Estrada.
\newblock On the thinness and proper thinness of a graph.
\newblock {\em Discrete Applied Mathematics}, 261:78--92, 2019.

\bibitem{MR4606115}
F.~Bonomo-Braberman and G.~A. Brito.
\newblock Intersection models and forbidden pattern characterizations for
  2-thin and proper 2-thin graphs.
\newblock {\em Discrete Applied Mathematics}, 339:53--77, 2023.

\bibitem{bolu76}
K.~S. Booth and G.~S. Lueker.
\newblock Testing for the consecutive ones property, interval graphs, and graph
  planarity using {PQ}-tree algorithms.
\newblock {\em Journal of Computer and System Sciences}, 13(3):335--379, 1976.

\bibitem{MR4249058}
S.~Chaplick, F.~V. Fomin, P.~A. Golovach, D.~Knop, and P.~Zeman.
\newblock Kernelization of graph {H}amiltonicity: proper {$H$}-graphs.
\newblock {\em SIAM Journal on Discrete Mathematics}, 35(2):840--892, 2021.

\bibitem{MR4332111}
S.~Chaplick, M.~T\"{o}pfer, J.~Voborn\'{\i}k, and P.~Zeman.
\newblock On {$H$}-topological intersection graphs.
\newblock {\em Algorithmica}, 83(11):3281--3318, 2021.

\bibitem{ourarx}
N.~Chiarelli, M.~Krnc, M.~Milani{\v c}, U.~Pferschy, N.~Piva{\v c}, and
  J.~Schauer.
\newblock Fair allocation of indivisible items with conflict graphs.
\newblock {\em Algorithmica}, 85:1459--1489, 2023.

\bibitem{corneil05}
D.~G. Corneil and U.~Rotics.
\newblock On the relationship between clique-width and treewidth.
\newblock {\em SIAM Journal on Computing}, 34(4):825--847, 2005.

\bibitem{MR3459604}
B.~Courcelle and I.~Durand.
\newblock Computations by fly-automata beyond monadic second-order logic.
\newblock {\em Theoretical Computer Science}, 619:32--67, 2016.

\bibitem{courengel93}
B.~Courcelle, J.~Engelfriet, and G.~Rozenberg.
\newblock Handle-rewriting hypergraph grammars.
\newblock {\em Journal of Computer and System Sciences}, 46(2):218--270, 1993.

\bibitem{MR1739644}
B.~Courcelle, J.~A. Makowsky, and U.~Rotics.
\newblock Linear time solvable optimization problems on graphs of bounded
  clique-width.
\newblock {\em Theory of Computing Systems}, 33(2):125--150, 2000.

\bibitem{MR1743732}
B.~Courcelle and S.~Olariu.
\newblock Upper bounds to the clique width of graphs.
\newblock {\em Discrete Applied Mathematics}, 101(1-3):77--114, 2000.

\bibitem{MR3380745}
M.~Cygan, F.~V. Fomin, {\L}.~Kowalik, D.~Lokshtanov, D.~Marx, M.~Pilipczuk,
  M.~Pilipczuk, and S.~Saurabh.
\newblock {\em Parameterized algorithms}.
\newblock Springer, 2015.

\bibitem{dallard2022computing}
C.~Dallard, F.~V. Fomin, P.~Golovach, T.~Korhonen, and M.~Milani{\v{c}}.
\newblock Computing tree decompositions with small independence number.
\newblock {\em arXiv:2206.15092}, 2022.

\bibitem{dallard2022firstpaper}
C.~Dallard, M.~Milani{\v{c}}, and K.~{\v{S}}torgel.
\newblock Treewidth versus clique number. {II}. {T}ree-independence number.
\newblock {\em arXiv:2111.04543}, 2022.

\bibitem{dallard2022secondpaper}
C.~Dallard, M.~Milani{\v{c}}, and K.~{\v{S}}torgel.
\newblock Treewidth versus clique number. {III}. {T}ree-independence number of
  graphs with a forbidden structure.
\newblock {\em arXiv:2206.15092}, 2022.

\bibitem{dpsw11}
A.~Darmann, U.~Pferschy, J.~Schauer, and G.~J. Woeginger.
\newblock Paths, trees and matchings under disjunctive constraints.
\newblock {\em Discrete Applied Mathematics}, 159:1726--1735, 2011.

\bibitem{MR4187151}
M.~de~Berg, H.~L. Bodlaender, S.~Kisfaludi-Bak, D.~Marx, and T.~C. van~der
  Zanden.
\newblock A framework for exponential-time-hypothesis-tight algorithms and
  lower bounds in geometric intersection graphs.
\newblock {\em SIAM Journal on Computing}, 49(6):1291--1331, 2020.

\bibitem{MR1097650}
D.~de~Werra.
\newblock Packing independent sets and transversals.
\newblock In {\em Combinatorics and graph theory}, volume~25 of {\em Banach
  Center Publ.}, pages 233--240. PWN, Warsaw, 1989.

\bibitem{Diaz2021}
J.~D{\'i}az, {\"O}.~Y. Diner, M.~Serna, and O.~Serra.
\newblock On list $k$-coloring convex bipartite graphs.
\newblock In C.~Gentile, G.~Stecca, and P.~Ventura, editors, {\em Graphs and
  Combinatorial Optimization: from Theory to Applications: CTW2020
  Proceedings}, pages 15--26. Springer, 2021.

\bibitem{DVORAK2012}
Z.~Dvořák and D.~Král’.
\newblock Classes of graphs with small rank decompositions are $\chi$-bounded.
\newblock {\em European Journal of Combinatorics}, 33(4):679--683, 2012.

\bibitem{MR1905626}
W.~Espelage, F.~Gurski, and E.~Wanke.
\newblock How to solve {NP}-hard graph problems on clique-width bounded graphs
  in polynomial time.
\newblock In {\em Graph-theoretic concepts in computer science}, volume 2204 of
  {\em Lecture Notes in Computer Science}, pages 117--128. Springer, 2001.

\bibitem{raey09}
G.~Even, M.~M. Halld{\'o}rsson, L.~Kaplan, and D.~Ron.
\newblock Scheduling with conflicts: online and offline algorithms.
\newblock {\em Journal of Scheduling}, 12(2):199--224, 2009.

\bibitem{felrosa09}
M.~R. Fellows, F.~A. Rosamond, U.~Rotics, and S.~Szeider.
\newblock Clique-width is {NP}-complete.
\newblock {\em SIAM Journal on Discrete Mathematics}, 23(2):909--939, 2009.

\bibitem{MR4276552}
F.~V. Fomin and P.~A. Golovach.
\newblock Subexponential parameterized algorithms and kernelization on almost
  chordal graphs.
\newblock {\em Algorithmica}, 83(7):2170--2214, 2021.

\bibitem{MR4141534}
F.~V. Fomin, P.~A. Golovach, and J.-F. Raymond.
\newblock On the tractability of optimization problems on {$H$}-graphs.
\newblock {\em Algorithmica}, 82(9):2432--2473, 2020.

\bibitem{MR4490048}
F.~V. Fomin and T.~Korhonen.
\newblock Fast {FPT}-approximation of branchwidth.
\newblock In {\em S{TOC} '22---{P}roceedings of the 54th {A}nnual {ACM}
  {SIGACT} {S}ymposium on {T}heory of {C}omputing}, pages 886--899. ACM, New
  York, 2022.

\bibitem{MR1973174}
M.~U. Gerber and D.~Kobler.
\newblock Algorithms for vertex-partitioning problems on graphs with fixed
  clique-width.
\newblock {\em Theoretical Computer Science}, 299(1-3):719--734, 2003.

\bibitem{MR936633}
M.~Gr\"{o}tschel, L.~Lov\'{a}sz, and A.~Schrijver.
\newblock {\em Geometric algorithms and combinatorial optimization}, volume~2
  of {\em Algorithms and Combinatorics: Study and Research Texts}.
\newblock Springer, 1988.

\bibitem{Gurski}
F.~Gurski.
\newblock A comparison of two approaches for polynomial time algorithms
  computing basic graph parameters.
\newblock {\em arXiv:0806.4073}, 2008.

\bibitem{MR4456180}
A.~Jacob, F.~Panolan, V.~Raman, and V.~Sahlot.
\newblock Structural parameterizations with modulator oblivion.
\newblock {\em Algorithmica}, 84(8):2335--2357, 2022.

\bibitem{MR4087200}
L.~Jaffke, O.-j. Kwon, and J.~A. Telle.
\newblock Mim-{W}idth {I}. {I}nduced path problems.
\newblock {\em Discrete Applied Mathematics}, 278:153--168, 2020.

\bibitem{MR4334540}
J.~Jeong, E.~J. Kim, and S.-i. Oum.
\newblock Finding branch-decompositions of matroids, hypergraphs, and more.
\newblock {\em SIAM Journal on Discrete Mathematics}, 35(4):2544--2617, 2021.

\bibitem{lipski1981efficient}
W.~Lipski and F.~P. Preparata.
\newblock Efficient algorithms for finding maximum matchings in convex
  bipartite graphs and related problems.
\newblock {\em Acta Informatica}, 15(4):329--346, 1981.

\bibitem{MR2294675}
C.~Mannino, G.~Oriolo, F.~Ricci, and S.~Chandran.
\newblock The stable set problem and the thinness of a graph.
\newblock {\em Operations Research Letters}, 35(1):1--9, 2007.

\bibitem{milanicRzazewski2022}
M.~Milani{\v{c}} and P.~Rz\c{a}\.{z}ewski.
\newblock Tree decompositions with bounded independence number: beyond
  independent sets.
\newblock {\em arXiv:2209.12315}, 2022.

\bibitem{mimt10}
A.~Muritiba, M.~Iori, E.~Malaguti, and P.~Toth.
\newblock {Algorithms for the bin packing problem with conflicts}.
\newblock {\em INFORMS Journal on Computing}, 22(3):401--415, 2010.

\bibitem{MR2479181}
S.-i. Oum.
\newblock Approximating rank-width and clique-width quickly.
\newblock {\em ACM Transactions on Algorithms}, 5(1):Art.~10, 20, 2009.

\bibitem{MR2232389}
S.-i. Oum and P.~Seymour.
\newblock Approximating clique-width and branch-width.
\newblock {\em Journal of Combinatorial Theory. Series B}, 96(4):514--528,
  2006.

\bibitem{pfsch09}
U.~Pferschy and J.~Schauer.
\newblock The knapsack problem with conflict graphs.
\newblock {\em Journal of Graph Algorithms and Applications}, 13(2):233--249,
  2009.

\bibitem{pfsch17}
U.~Pferschy and J.~Schauer.
\newblock Approximation of knapsack problems with conflict and forcing graphs.
\newblock {\em Journal of Combinatorial Optimization}, 33(4):1300--1323, 2017.

\bibitem{MR4433312}
A.~Rafiey.
\newblock Recognizing interval bigraphs by forbidden patterns.
\newblock {\em Journal of Graph Theory}, 100(3):504--529, 2022.

\bibitem{MR2323400}
M.~Rao.
\newblock M{SOL} partitioning problems on graphs of bounded treewidth and
  clique-width.
\newblock {\em Theoretical Computer Science}, 377(1-3):260--267, 2007.

\bibitem{Santos19}
L.~F.~M. Santos, R.~S. Iwayama, L.~B. Cavalcanti, L.~M. Turi, F.~E.
  de~Souza~Morais, G.~Mormilho, and C.~B. Cunha.
\newblock A variable neighborhood search algorithm for the bin packing problem
  with compatible categories.
\newblock {\em Expert Systems with Applications}, 124:209--225, 2019.

\bibitem{MR4174126}
A.~Scott and P.~Seymour.
\newblock A survey of {$\chi$}-boundedness.
\newblock {\em Journal of Graph Theory}, 95(3):473--504, 2020.

\bibitem{MR1852483}
K.~Skodinis.
\newblock Efficient analysis of graphs with small minimal separators.
\newblock In {\em Graph-theoretic concepts in computer science}, volume 1665 of
  {\em Lecture Notes in Computer Science}, pages 155--166. Springer, 1999.

\bibitem{MR1971502}
J.~P. Spinrad.
\newblock {\em Efficient graph representations}, volume~19 of {\em Fields
  Institute Monographs}.
\newblock American Mathematical Society, Providence, RI, 2003.

\bibitem{MR3775804}
N.~Yolov.
\newblock Minor-matching hypertree width.
\newblock In {\em Proceedings of the {T}wenty-{N}inth {A}nnual {ACM}-{SIAM}
  {S}ymposium on {D}iscrete {A}lgorithms}, pages 219--233. SIAM, Philadelphia,
  PA, 2018.

\end{thebibliography}

\end{document}